\RequirePackage{amsthm}
\documentclass[sn-mathphys-num]{sn-jnl}

\usepackage{graphicx}\usepackage{bmpsize}
\usepackage{amsmath,amssymb,amsfonts}\usepackage{mathrsfs}\usepackage[title]{appendix}\usepackage{xcolor}\usepackage{textcomp}\usepackage{manyfoot}\usepackage{bookmark}\usepackage{listings}\usepackage{lmodern}\usepackage{enumitem}

\usepackage[T1]{fontenc}
\usepackage{ascmac}
\usepackage{booktabs}
\usepackage{array}
\usepackage[algoruled,linesnumbered,algo2e,vlined]{algorithm2e}
\usepackage{url}
\usepackage{multirow}
\usepackage{multicol}
\usepackage{cases}
\usepackage[normalem]{ulem}
\usepackage{boites}
\usepackage{here}
\usepackage{colortbl}
\usepackage{color}
\graphicspath{{./figure/}}

\usepackage{physics} 

\usepackage[pagewise,mathlines]{lineno}

\usepackage[final,mode=multiuser]{fixme}
\fxsetup{theme=color}
\definecolor{fxtarget}{rgb}{0.0000,0.0000,0.4823}
\fxuseenvlayout{colorsig}
\fxusetargetlayout{color}
\fxuseenvlayout{colorsig}
\fxusetargetlayout{color}
\FXRegisterAuthor{tt}{att}{TT}
\FXRegisterAuthor{yy}{ayy}{YY}
\FXRegisterAuthor{hs}{ahs}{HS}
\FXRegisterAuthor{tm}{atm}{TM}
\FXRegisterAuthor{yn}{ayn}{YN}
\FXRegisterAuthor{si}{asi}{SI}
\fxsetface{env}{}

\newcommand{\mt}[1]{\mathtt{#1}}
\newcommand{\str}[1]{\langle #1 \rangle}

\newcommand{\minid}{\mathsf{min\_id}}
\newcommand{\CT}{\mathsf{CT}}
\newcommand{\ctlcs}{\mathsf{ct\_lcs}}

\newcommand{\minst}{\mathsf{min}}

\newcommand{\cand}{\mathsf{cand}}
\newcommand{\lnd}{\mathit{LND}}
\newcommand{\ND}{{\mathit{ND}}}
\newcommand{\Lw}{L_w}
\newcommand{\lcs}{\mathsf{LCS}}

\newcommand{\nonl}{\renewcommand{\nl}{\let\nl\oldnl}}

\newcommand{\zero}{\vb{0}}

\newcommand{\CTLCSprob}{\textsf{CT-LCS}}
\newcommand{\CTMSeqprob}{\textsf{CT-MSeq}}
\newcommand{\SCTMSeqprob}{\textsf{s-CT-MSeq}}

\theoremstyle{thmstyleone}\newtheorem{theorem}{Theorem}\newtheorem{lemma}{Lemma}\newtheorem{corollary}{Corollary}\newtheorem{claim}{Claim}

\theoremstyle{thmstyletwo}\newtheorem{problem}{Problem}\newtheorem{observation}{Observation}

\theoremstyle{thmstylethree}\newtheorem{definition}{Definition}

\begin{document}

\title[Subsequence Matching and LCS under Cartesian-Tree Equivalence]{Subsequence Matching and LCS under Cartesian-Tree Equivalence}

\author[1]{\fnm{Taketo} \sur{Tsujimoto}}\email{tsujimoto.taketo.852@s.kyushu-u.ac.jp}

\author*[1]{\fnm{Yuki} \sur{Yonemoto}}\email{yonemoto.yuuki.240@s.kyushu-u.ac.jp} 

\author[2]{\fnm{Hiroki} \sur{Shibata}}\email{shibata.hiroki.753@s.kyushu-u.ac.jp} 

\author[3]{\fnm{Takuya} \sur{Mieno}}\email{tmieno@uec.ac.jp} 

\author[4]{\fnm{Yuto} \sur{Nakashima}}\email{nakashima.yuto.003@m.kyushu-u.ac.jp} 

\author[4]{\fnm{Shunsuke} \sur{Inenaga}}\email{inenaga.shunsuke.380@m.kyushu-u.ac.jp} 

\affil[1]{\orgdiv{Department of Information Science and Technology, Kyushu University}, \country{Japan}}

\affil[2]{\orgdiv{Joint Graduate School of Mathematics for Innovation, Kyushu University}, \country{Japan}}

\affil[3]{\orgdiv{Department of Computer and Network Engineering, University of Electro-Communications}, \country{Japan}}

\affil[4]{\orgdiv{Department of Informatics, Kyushu University}, \country{Japan}}

\abstract{
  Two strings of the same length are said to \emph{Cartesian-tree match}~(\emph{CT-match})
  if their Cartesian-trees are isomorphic [Park et al., TCS 2020].
  Cartesian-tree matching is a natural model that allows for capturing
  similarities of numerical sequences.
  Oizumi et al. [CPM 2022] showed that subsequence pattern matching
  under CT-matching model (\emph{CT-MSeq}) can be solved in $O(nm \log \log n)$ time,
  where $n$ and $m$ are text and pattern lengths, respectively.
  This current article follows this line of research,
  and gives the following new results:\\
  (1) An $O(nm)$-time CT-MSeq algorithm for binary alphabets.\\
  (2) An $O((nm)^{1-\epsilon})$-time conditional lower bound for the CT-MSeq problem on alphabets of size 4,
  for any constant $\epsilon > 0$, under the Orthogonal Vector Hypothesis (OVH).\\
  Further, we introduce the new problem of 
  \emph{longest common subsequence under CT-matching} (\emph{CT-LCS})
  for two given strings $S$ and $T$ of length $n$,
  and present the following results:\\
  (3) An $O(n^6)$-time CT-LCS algorithm for general ordered alphabets.\\
  (4) An $O(n^2 / \log n)$-time CT-LCS algorithm for binary alphabets.\\
  (5) An $O(n^{2-\epsilon})$-time conditional lower bound
  for the CT-LCS problem on alphabets of size 5, for any constant $\epsilon > 0$, under OVH.
}


\maketitle

\section{Introduction}

\subsection{Cartesian-tree equivalence of strings}

In the recent large-scale data era,
efficiently-computable string matching models
which are capable of capturing \emph{structural similarities}
are important for the analysis of numerical sequences such as time series.
For instance, in financial markets,
it seems more natural to focus on
price fluctuation patterns in stock data,
than to look for identical matches.

\emph{Order-preserving matching} (\emph{OP-matching})~\cite{kim:eades:etal:park:tokuyama:tcs2014oppm,kubica2013linear} is a natural model for dealing with numerical sequences: Two strings $A$ and $B$ of length $n$ are said to be OP-match (or OP-equivalent) iff the lexicographical rank of $A[i]$ in $A$ and that of $B[i]$ in $B$ are equal for all $1 \leq i \leq n$.

\emph{Cartesian-tree matching} (\emph{CT-matching}), first proposed by Park et al.~\cite{ParkBALP20}, is another model for dealing with numerical sequences:
Two strings $A$ and $B$ of length $n$ are said to be CT-match (or CT-equivalent) iff the (unlabeled) Cartesian-trees~\cite{gabow1984scaling} of $A$ and $B$ are isomorphic.
The CT-matching model is a relaxation of the OP-matching model,
i.e., any OP-matching strings also CT-match,
but the opposite is not true
(for instance, $A = \str{\mt{1},\mt{1},\mt{2}}$ and $B = \str{\mt{1},\mt{1},\mt{1}}$ CT-match, but they do not OP-match).
CT-matching has attracted attention in terms of pattern matching~\cite{ParkBALP20,SongGRFLP21},
string periodicity~\cite{ParkBALP20}, and indeterminate strings~\cite{GawrychowskiGL20}.

\subsection{Subsequence matching under Cartesian-tree equivalence}

In this article, we first deal with \emph{subsequence} matching under the CT-matching model (\emph{CT-MSeq}):
Given a text $T$ of length $n$ and a pattern $P$ of length $m$,
find all minimal intervals $[i,j]$ in $T$
such that $T[i..j]$ contains a subsequence $Q$ that CT-matches $P$.
This can be seen as the CT-matching version of the classical \emph{episode matching}~\cite{DasFGGK97}.
Oizumi et al.~\cite{OizumiKMIA22} showed an algorithm
that solves the CT-MSeq problem in $O(nm \log \log n)$ time and $O(n \log m)$ space
for general ordered alphabets of arbitrary size.
This is interesting since 
subsequence matching is NP-hard under OP-matching~\cite{bose:ipl1998lis:npc:pattern}.

We revisit the CT-MSeq problem and present the following new results in the case of smaller alphabets:
\begin{enumerate}
\item[(1)] an $O(nm)$-time and $O(n)$-space CT-MSeq algorithm for binary alphabets.
\item[(2)] an $O((nm)^{1-\epsilon})$-time conditional lower bound for the CT-MSeq problem on alphabets of size 4, for any constant $\epsilon > 0$,
  under the Orthogonal Vector Hypothesis (OVH).
\end{enumerate}

We achieve Result (1) with the $O(nm)$-time solution
in the binary case
by exploiting interesting properties of CT-matching on binary strings.
Result (2), which extends the quadratic conditional lower bound
for episode matching~\cite{BilleGMSW22},
implies that Oizumi et al.'s $O(nm \log \log n)$-time solution for the CT-MSeq problem~\cite{OizumiKMIA22} is optimal up to logarithmic factors,
unless OVH fails. 
\tmnote*{modified}{We remark that OVH is implied by the famous Strong Exponential Time Hypothesis (SETH)~\cite{BringmannK15}.
}

\subsection{LCS under Cartesian-tree model}

We then extend our research to the well-studied class of string problem,
\emph{longest common subsequences} (\emph{LCSs}),
that is one of the most fundamental models for measuring string similarities.

It is well known that, under the exact matching model,
(the length of) an LCS of two given strings $S$ and $T$ of length $n$ can be computed in $O(n^2)$ time and space by standard dynamic programming, or in $O(n^2 / \log n)$ time and space~\cite{MASEK198018} by the so-called ``Four-Russians'' method in the word RAM~\cite{ArlazarovEtAl1970}.
These quadratic and weakly subquadratic time bounds are believed to be essentially optimal,
since a strongly subquadratic $O(n^{2 - \epsilon})$-time solution to LCS with any constant $\epsilon > 0$ refutes
OVH.
Indeed, while there are a number of algorithms for computing LCS
whose running times are dependent on other parameters (e.g.~\cite{HuntS77,NakatsuKY82,ApostolicoBG92,Sakai12}), their worst-case time complexities remain $\Omega(n^2)$.

As previously stated, subsequence matching under OP-matching is NP-hard.
It is thus immediate that \emph{order-preserving longest common subsequence} (\emph{OP-LCS}) is also NP-hard.

These arguments pose the following natural question - 
Is the \emph{CT-LCS problem} also polynomial-time solvable?
Here, the CT-LCS problem is, given two strings $S$ and $T$ of length $n$,
to compute (the length) of a longest string $Q$
such that both $S$ and $T$ have subsequences that CT-match $Q$. 
We answer this question affirmatively, by presenting the following efficient algorithms and conditional lowerbound:
\begin{enumerate}
\item[(3)] an $O(n^6)$-time and $O(n^4)$-space CT-LCS algorithm for general ordered alphabets.
\item[(4)] an $O(n^2 / \log n)$-time and space CT-LCS algorithm
  for binary alphabets.
\item[(5)] an $O(n^{2-\epsilon})$-time conditional lower bound on the CT-LCS problem for alphabets of size 5, for any constant $\epsilon > 0$,
  under OVH.
\end{enumerate}  
  
While the $O(n^6)$-time solution in the general case (Result (3)) is based on
the idea of \emph{pivoted Cartesian-trees} from Oizumi et al.~\cite{OizumiKMIA22},
our $O(n^2 / \log n)$-time solution for the binary case (Result (4))
is built on a completely different approach that exploits interesting properties of CT-matching on binary strings, which we also use for the CT-MSeq problem.

\subsection{Organization}
The rest of this paper is organized as follows:
Section~\ref{sec:pre} gives basic notions.
In Section~\ref{sec:binary_property},
we present combinatorial properties of CT-matching on binary strings,
which will play central roles
for our CT-MSeq and CT-LCS algorithms presented respectively in
Sections~\ref{sec:binary_matching} and~\ref{sec:binary_LCS}.
\yynote*{changed}{In Section~\ref{sec:alg_on_binary_seq}, we present $O(nm)$-time solution for the CT-Mseq problem in the case of binary alphabets
  and our quadratic conditional lower bound for the CT-MSeq problem.
  In section~\ref{sec:ctlcs}, we present our $O(n^6)$-time solution and $O(n^2)$-time solution for the CT-LCS problem in the case of general ordered alphabets and binary alphabets, respectively,
  and quadratic conditional lower bounds for the CT-LCS problem.
}We conclude in Section~\ref{sec:concl}.

A preliminary version of this paper appeared in~\cite{TsujimotoSMNI24}.
The new materials in this full version are the following:
\begin{itemize}
\item Result (1): $O(nm)$-time CT-MSeq algorithm for binary alphabets.
\item Result (2): $O((nm)^{1-\epsilon})$-time conditional lower bound for the CT-MSeq problem.
\item Result (5): $O(n^{2-\epsilon})$-time conditional lower bound for the CT-LCS problem.  
\end{itemize}
 \section{Preliminaries}
\label{sec:pre}

\subsection{Basic notations: strings and vectors}
For any positive integer $i$, we define a set $[i] = \{1,\ldots, i\}$ of $i$ integers.
Let $\Sigma$ be an \emph{ordered alphabet} of size $\sigma$.
For simplicity, let $\Sigma = \{0, \ldots, \sigma-1\}$.
An element of $\Sigma$ is called a \emph{character}.
A sequence of characters is called a \emph{string}.
The \emph{length} of string $S$ is denoted by $|S|$.
The empty string $\varepsilon$ is the string of length $0$.
If $S = XYZ$, then $X$, $Y$, and $Z$ are respectively called
a \emph{prefix}, \emph{substring}, and \emph{suffix} of $S$.
For a string $S$, $S[i]$ denotes the $i$-th character of $S$ for each $i$ with $1 \le i \le |S|$.
For each $i,j$ with $1 \le i \le j \le |S|$, $S[i..j]$ denotes the substring of $S$ that begins at position $i$ and ends at position $j$.
For convenience, let $S[i.. j] = \varepsilon$ for $i > j$.
We write $\minst(S) = \min \{S[i] \;|\; i \in [n]\}$ for the minimum character contained in the string $S$.
For any $0 \le m \le n$, let  $\mathcal{I}_m^n$ be the set consisting of all \emph{subscript sequence} $I = (i_1, \ldots, i_m) \in [n]^m$
in ascending order satisfying $1 \le i_1 < \cdots < i_m \le n$.
For subscript sequence $I = (i_1, \ldots, i_m) \in \mathcal{I}_m^n$, we denote by $S_I = \str{S[i_1], \ldots, S[i_m]}$ the \emph{subsequence} of $S$ corresponding to $I$.
For a subscript sequence $I$ and its elements $i_s, i_t \in I$ with $i_s \le i_t$, 
$I[i_s:i_t]$ denotes  
the substring
of $I$ that starts with $i_s$ and ends with $i_t$.
For a sequence $T_1, \ldots, T_k$ of strings, we define $\prod_{i = 1}^{k} T_i = T_1\cdots T_k$.

For a $d$-dimensional binary vector $\alpha \in \{0, 1\}^d$ and an integer $1 \le t \le d$,
the $t$-th element in $\alpha$ is denoted by $\alpha[t]$.
For two vectors $\alpha, \beta \in \{0, 1\}^d$, they are said to be \emph{orthogonal}
if $\sum^d_{i=1}{\alpha[i]\beta[i]} = 0$ holds.

\subsection{Subsequence matching and LCS under Cartesian-tree equivalence}
For a string $S$, let $\minid(S)$ denote the least index $i$
such that $S[i]$ is the smallest element in $S$.
\begin{definition}[Cartesian-tree]
The \emph{Cartesian-tree} of a string $S$, denoted by $\CT(S)$, is the ordered binary tree recursively defined as follows:
If $S = \varepsilon$, then $\CT(S)$ is an empty tree, and otherwise,
$\CT(S)$ is the tree rooted at $i = \minid(S)$ such that
the left subtree of $i$ is $\CT(S[1.. i-1])$ and
the right subtree of $i$ is $\CT(S[i+1.. |S|])$.
 \label{def:ct}
\end{definition}
\yynote*{Added}{Fig.~\ref{fig:example_ct} shows $\CT(S)$ for a string $S=\str{\mt{5},\mt{8},\mt{2},\mt{6},\mt{4},\mt{7},\mt{4}}$ and $\CT(S_b)$ for a binary string $S_b=\str{\mt{1},\mt{0},\mt{0},\mt{1},\mt{0},\mt{1},\mt{0}}$.
}\yynote*{Added}{As shown in Fig.~\ref{fig:example_ct}, 
  the Cartesian tree recursively splits the string based on the smallest character without considering the ordering of characters between the left and right subtrees, 
  thereby expressing a flexible ordering structure of the characters.
}

\begin{figure}[h]
  \centering
  \includegraphics[keepaspectratio, scale=0.4]{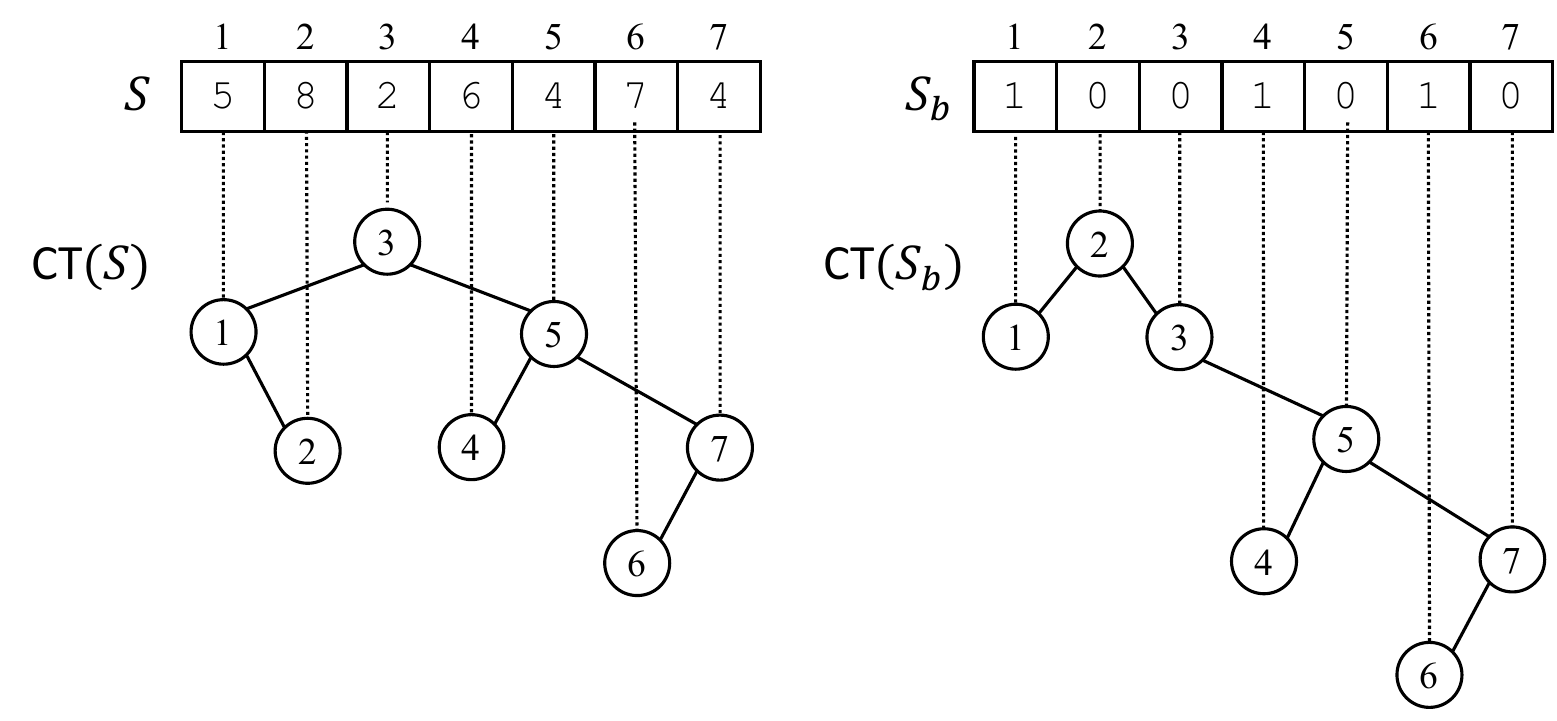}
  \caption{This figure shows examples of Cartesian-trees $\CT(S)$ for a string $S=\str{\mt{5},\mt{8},\mt{2},\mt{6},\mt{4},\mt{7},\mt{4}}$, 
  $\CT(S_b)$ for a binary string $S_b=\str{\mt{1},\mt{0},\mt{0},\mt{1},\mt{0},\mt{1},\mt{0}}$. 
  Each node in $\CT(S)$ and $\CT(S_b)$ are labeled by the corresponding position in $S$.
  Since $\minid(S_b)$ is the least index $i$ such that $S_b[i]$ is the smallest element in $S_b$,
  $\CT(S_b)$ is the tree rooted at $i = 2$.
  }
  \label{fig:example_ct}
\end{figure}

For two strings $S$ and $T$ of equal length,
the two Cartesian-trees $\CT(S)$ and $\CT(T)$ are \textit{isomorphic}
if they have the same tree topology as ordered trees~\cite{hoffmann1982pattern}, which are tree data structures where the children of each node are ordered.
We denote it by $\CT(S) = \CT(T)$.
We say that two strings $S$ and $T$ \emph{CT-match} if $\CT(S) = \CT(T)$.

\begin{definition}[CT-subsequence]
A string $P$ is said to be a \emph{CT-subsequence} of a string $T$
if there is a subsequence $Q$ of $T$ such that $\CT(Q) = \CT(P)$.
\end{definition}

\vspace{-8mm}

\begin{definition}[CT-occurrence interval]
    An interval $[\ell, r]$ in $T$ is said to be
  a \emph{CT-occurrence interval} for $P$
  if there is a subsequence $T[i_1] \cdots T[i_{|P|}]$ of $T[\ell..r]$
  such that $\CT(T[i_1] \cdots T[i_{|P|}]) = \CT(P)$,
  where $\ell \leq i_1 \leq \cdots \leq i_{|P|} \leq r$.
\end{definition}
A CT-occurrence interval $[\ell, r]$ for $P$ in $T$
is said to be \emph{minimal}
if no proper sub-interval of $[\ell,r]$ is a CT-occurrence interval for $P$ in $T$.  

Our first problem is the following:
\begin{problem}[CT-subsequence matching problem ({\CTMSeqprob})~\cite{OizumiKMIA22}] \label{prob:stmseq}
  Given two strings $T$ and $P$, 
  find all minimal CT-occurrence intervals for $P$ in $T$.
  \label{pro:CT-MSeq}
\end{problem}

\yynote{Added}{Fig.~\ref{fig:ct_matching} shows an example of CT-subsequence matching.
}\begin{figure}[h]
  \centering
  \includegraphics[keepaspectratio, scale=0.4]{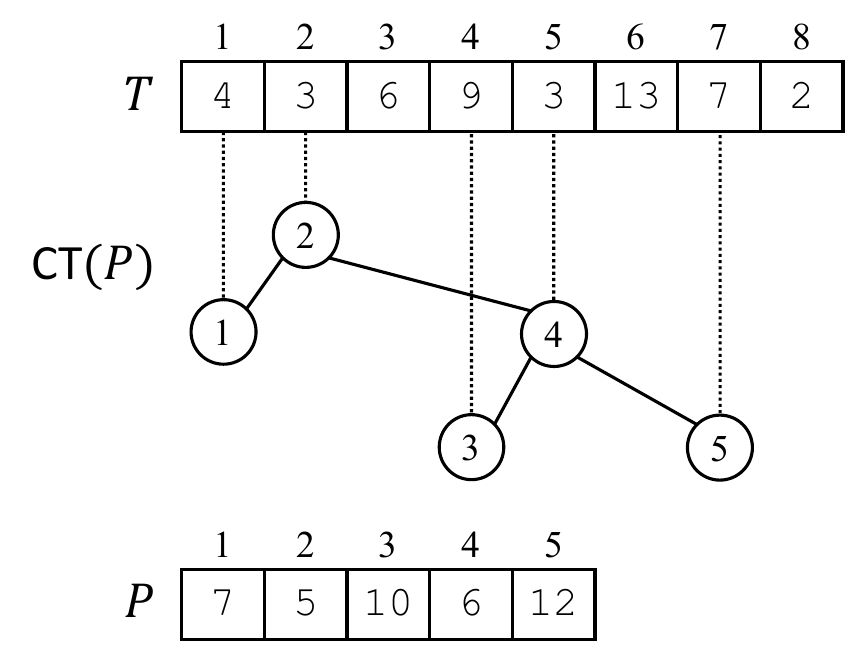}
  \caption{$P=\str{\mt{7},\mt{5},\mt{10},\mt{6},\mt{12}}$ is a CT-subsequence of $T=\str{\mt{4}, \mt{3}, \mt{6}, \mt{9}, \mt{6}, \mt{13}, \mt{7}, \mt{2}}$. Each node in $\CT(P)$ is labeled by the corresponding position in $P$.}
  \label{fig:ct_matching}
\end{figure}

A string $Q$ is said to be a \emph{common CT-subsequence} of two strings $S$
and $T$ if $Q$ is a CT-subsequence of both $S$ and $T$.
A string $Q$ is said to be a \emph{longest common CT-subsequence} (\emph{CT-LCS}) of $S$ and $T$ if there are no common CT-subsequences of $S$ and $T$ longer than $Q$.
We show an example of CT-LCS in Fig.~\ref{fig:ctlcs}.
The length of CT-LCS of strings $S$ and $T$ is denoted by $\ctlcs(S,T)$.

Our second problem is the following:
\begin{problem}[Longest common CT-subsequence problem ({\CTLCSprob})]
  Given two strings $S$ and $T$, compute $\ctlcs(S,T)$.
  \label{pro:CT-LCS}
\end{problem}

\begin{figure}[h]
  \centering
  \includegraphics[keepaspectratio, scale=0.4]{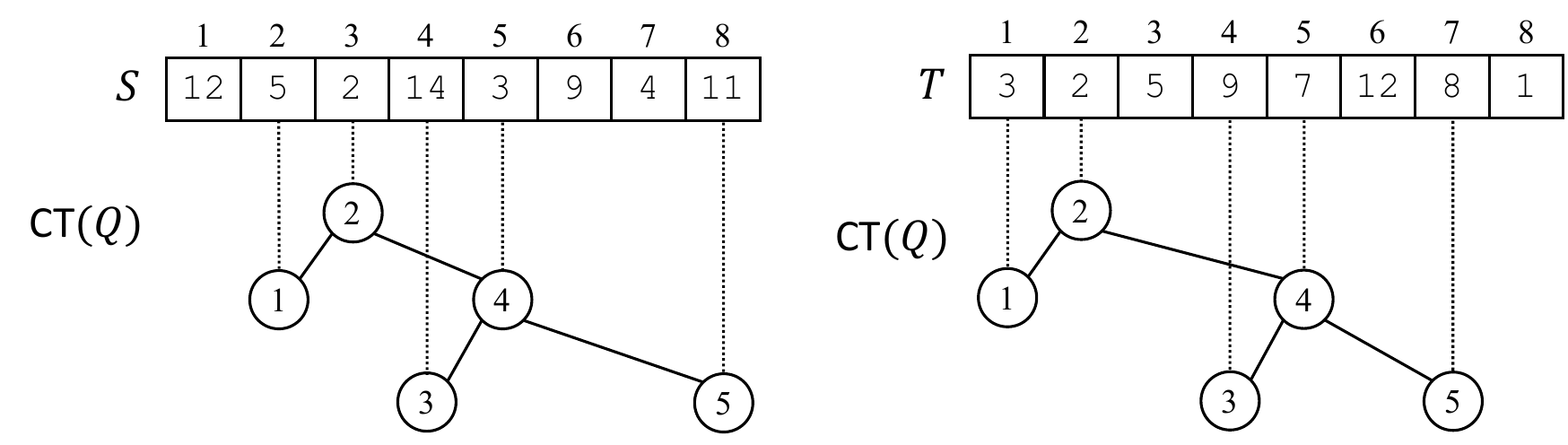}
  \caption{$Q=\str{\mt{5},\mt{2},\mt{14},\mt{3},\mt{11}}$ is a longest common CT-subsequence of $S=\str{\mt{12}, \mt{5}, \mt{2}, \mt{14}, \mt{3}, \mt{9}, \mt{4}, \mt{11}}$ and $T=\str{\mt{3}, \mt{2}, \mt{5}, \mt{9}, \mt{7}, \mt{12}, \mt{8}, \mt{1}}$. Each node in $\CT(Q)$ is labeled by the corresponding position in $Q$.}
  \label{fig:ctlcs}
\end{figure}

\subsection{Hardness assumptions}\label{sec:hardness}

Let us consider the well-known \emph{$k$-SAT problem}:
Given a propositional logic formula of conjunctive normal form (CNF) which has $n$ variables and at most $k$ literals in each clause,
determine whether there exists an interpretation that satisfies the input formula.
Impagliazzo and Paturi~\cite{ImpagliazzoP01} introduces the Strong Exponential Time Hypothesis (SETH),
which states that there is no algorithm that solves $k$-SAT inherently faster than exhaustive search.
\begin{itemize}
  \item {\bf Strong Exponential Time Hypothesis (SETH):}
For any $\epsilon > 0$, there exists $k \geq 3$ such that the $k$-SAT problem cannot be solved in $O(2^{(1-\epsilon)n})$ time.
\end{itemize}

Further, consider the Orthogonal Vector (OV) problem:
Given two sets of $n_A$ vectors $A$ and $n_B$ vectors $B$ such that $A,B \subset \{ \alpha \mid \alpha \in \{0,1\}^d \}$,
  determine whether there exist two vectors $\alpha, \beta \in A$ that are orthogonal.
Regarding the OV problem, the following hypotheses are considered in the literature:
\begin{itemize}
\item {\bf Orthogonal Vector Hypotheses (OVH):}
    For $n_A = n_B$, there are no $\epsilon > 0$ such that the OV problem can be solved in $O({n_A}^{2-\epsilon} \mathsf{poly}(d))$ time.
  \item {\bf Unbalanced Orthogonal Vector Hypotheses (UOVH):}
    Let $0 < \lambda \le 1$, there are no $\epsilon > 0$ such that the OV problem restricted to $n_B = \Theta({n_A}^\lambda)$ and $d \le n^{o(1)}$
    can be solved in $O((n_An_B)^{1-\epsilon})$ time.
\end{itemize}
It is known that SETH implies OVH and UOVH~\cite{williams2005new,BringmannK15}.
Thus, if OVH or UOVH is false, then SETH is also false.
 \section{Combinatorial properties of CT-matching on binary strings}
\label{sec:binary_property}

In this section, we exploit new properties of
CT-matching on binary strings,
which we will use for designing efficient algorithms
for {\CTMSeqprob} and {\CTLCSprob} on binary strings in Section~\ref{sec:binary_matching}
and Section~\ref{sec:binary_LCS},
respectively.

We first recall the \emph{parent-distance representation} presented by Park et al.~\cite{ParkBALP20}:
Given a string $S[1..n]$, the \textup{parent-distance representation} of $S$ is
  an integer string $\mathit{PD}(S)[1..n]$, which is defined as follows:
  \begin{linenomath}
    \begin{equation*}
      \mathit{PD}(S)[i] =
      \begin{cases} 
        i - \max_{1 \leq j < i}\{j \;|\; S[j] \leq S[i]\} \quad & \text{\rm{if such} } j \text{ \rm{exists}},\\   
        0                                               & \text{\rm{otherwise}}.
      \end{cases}
    \end{equation*}      
  \end{linenomath}

For example, the parent-distance representation of string $S = \str{\mt{1},\mt{0},\mt{1},\mt{1},\mt{0},\mt{0},\mt{1}}$
is $\mathit{PD}(S) = \str{\mt{0},\mt{0},\mt{1},\mt{1},\mt{3},\mt{1},\mt{1}}$.

\begin{lemma}[\cite{ParkBALP20}] \label{thm:pd_and_ct}
  Two strings $S_1$ and $S_2$ CT-match if and only if $S_1$ and $S_2$ have the same
  parent-distance representations.
\end{lemma}

Lemma~\ref{thm:pd_and_ct} allows for determining whether two strings CT-match or not.
We will only use this representation to guarantee the correctness of our algorithms for the binary case, and do not explicitly compute it.

We start from a simple observation of CT-matching on binary strings.
The following observation and lemmas will support our algorithm for the binary alphabet $\{0, 1\}$.
\begin{observation} \label{obs:nondecreasing}
  For any non-empty string $S$,
  $\mathit{PD}(S) = 01^{|S|-1}$ if and only if $S$ is non-decreasing.
  Namely, a non-decreasing sequence CT-matches only a non-decreasing sequence.
\end{observation}
Observation~\ref{obs:nondecreasing} implies the following lemma:
\begin{lemma} \label{lem:only1_ct_match}
  Let $S_1 = 1^n$ and $S_2$ be a binary string of length $n$.
Then, $S_1$ and $S_2$ CT-match if and only if $S_2 = 0^i1^{n-i}$ for some integer $i \geq 0$.
\end{lemma}
From now on, we discuss the case where $0$ appears in both $S_1$ and $S_2$.

\begin{lemma}
\label{lem:bi_ct_match}
  For two binary strings $S_1$ and $S_2$ of length $n$ both containing $0$, 
  $S_1$ and $S_2$ CT-match if and only if there exist a string $w$ and two integers $i \geq 1$, $j \geq 1$ 
  such that $S_1 = w0^i1^{n-|w|-i}$ and $S_2 = w0^j1^{n-|w|-j}$.
\end{lemma}

\begin{proof}
  ($\Longrightarrow$) Since $S_1$ and $S_2$ CT-match, $\mathit{PD}(S_1)=\mathit{PD}(S_2)$ (by Theorem~\ref{thm:pd_and_ct}).
  Let $p$ (resp., $q$) be the smallest integer
  such that $S_1[p..n] = 0^i1^j$ (resp., $S_2[q..n] = 0^i1^j$) for some $i \geq 1$, $j \geq 0$.
  Assume that $1 \leq q < p$.
  Since $S_1[p-1]=1$ and $S_1[p]=0$, either $\mathit{PD}(S_1)[p]=0$ or $\mathit{PD}(S_1)[p]>1$ holds 
  (i.e., $\mathit{PD}(S_1)[p] \neq 1$).
  On the other hand, $\mathit{PD}(S_2)[p] = 1$ holds since $S_2[p-1] \leq S_2[p]$.
  Then $\mathit{PD}(S_1)[p] \neq \mathit{PD}(S_2)[p]$, which is a contradiction.
  By a similar discussion, $1 \leq p < q$ also leads a contradiction.
  Next, we assume that $p = q = 1$.
  This assumption implies that the statement holds since $w = \varepsilon$.
  Assume that $1 < p = q$.
  Suppose on the contrary that $S_1[1..p-1] \neq S_2[1..p-1]$.
  There exists an integer $j^*$ such that $j^* = \max_{1 \leq j \leq p-1}\{j \mid S_1[j] \neq S_2[j]\}$.
  This means that $S_1[j^*+1..p-1] = S_2[j^*+1..p-1]$.
  We assume w.l.o.g. that $S_1[j^*] = 0$ and $S_2[j^*] = 1$
  (the other case is symmetric).
  \begin{itemize}
    \item If $0$ does not appear in $S_1[j^*+1..p-1]$, $\mathit{PD}(S_1)[p] = p-j^* > 0$ holds.
    On the other hand, either $\mathit{PD}(S_2)[p] = 0$ or $\mathit{PD}(S_2)[p] > p-j^*$ holds.
    Thus $\mathit{PD}(S_1) \neq \mathit{PD}(S_2)$, which is a contradiction.
  \item If $0$ appears in $S_1[j^*+1..p-1]$, 
    there exists an integer $i^*$ such that
    $i^{*} = \min_{j^*+1 \leq j \leq p-1}\{j \mid S_1[j] = S_2[j] = 0\}$ holds.
    This implies that $\mathit{PD}(S_1)[i^{*}]=i^{*}-j^* > 0$ holds.
    On the other hand, either $\mathit{PD}(S_2)[i^*] = 0$ or $\mathit{PD}(S_2)[i^*] > i^*-j^*$ holds.
    Thus $\mathit{PD}(S_1) \neq \mathit{PD}(S_2)$, which is a contradiction.
  \end{itemize}
  Therefore, $S_1 = w0^i1^{n-p+1-i}$ and $S_2 = w0^j1^{n-p+1-j}$ hold for some integers $i, j$
  where $w =S_1[1..p-1]=S_2[1..p-1]$.

  \noindent ($\Longleftarrow$) If $w = \varepsilon$, it is clear that $S_1$ and $S_2$ CT-match.
  We consider the case of $w \neq \varepsilon$.
  We show that $\mathit{PD}(S_1) = \mathit{PD}(S_2)$ holds, which is suffice due to Lemma~\ref{thm:pd_and_ct}.
It is easy to see that $\mathit{PD}(S_1)[i] = \mathit{PD}(S_2)[i]$ for all $i$ that satisfies $1 \leq i \leq |w|$.
  Moreover, $\mathit{PD}(S_1)[i] = \mathit{PD}(S_2)[i] = 1$ also holds for all $i$ with $|w|+1 < i \leq n$.
  If $0$ does not appear in $w$, $\mathit{PD}(S_1)[|w|+1] = \mathit{PD}(S_2)[|w|+1] = 0$.
  We assume that $0$ appears in $w$ for the remaining case.
  Let $j^* = \max_{1 \leq j \leq |w|}\{j \mid w[j]=0\}$.
  Since $S_1[|w|+1] = S_2[|w|+1] = 0$ and $S_1[1..|w|] = S_2[1..|w|]$,
  $\mathit{PD}(S_1)[|w|+1] = \mathit{PD}(S_2)[|w|+1]$ holds. 
  Therefore $\mathit{PD}(S_1) = \mathit{PD}(S_2)$.
\end{proof}
 \section{CT-subsequence matching problems}
\label{sec:alg_on_binary_seq}

We propose a quadratic algorithm for solving {\CTMSeqprob} for the binary alphabet $\{0, 1\}$ in Section~\ref{sec:binary_matching}, and
give a conditional lower bound for {\CTMSeqprob} for a general ordered alphabet of size four in Section~\ref{sec:lb_ctmseq}.
Throughout this section, 
we assume that the text $T$ of length $n$ and the pattern $P$ of length $m \leq n$ are binary strings
and discard the assumption that
all characters are distinct in $T$ and in $P$.

\subsection{Algorithm for binary {\CTMSeqprob}}
\label{sec:binary_matching}
Our algorithm for {\CTMSeqprob} on binary strings is based on Lemma~\ref{lem:only1_ct_match} and Lemma~\ref{lem:bi_ct_match}.
Firstly, we present some definitions to describe our algorithm.
For any integers $\ell \in [m]$ and $j \in [n]$,
let $\ND(\ell, j)$ be the largest integer $k$ such that
$T[k..j]$ contains a non-decreasing subsequence of length $\ell$
and $\ND_0(\ell, j)$ be the largest integer $k$ such that
$T[k..j]$ contains a non-decreasing subsequence of length $\ell$ \emph{including at least one $0$'s}.
Further, for a string $w$ and integers $i, j$ with $i\in [|w|]$ and $j\in [n]$,
let $\Lw(i,j)$ be the largest integer $k$ such that $T[k..j]$ contains $w[1..i]$ as a subsequence.
If no such $k$ exists, let the values of $\ND(\ell, j)$, $\ND_0(\ell, j)$, and $\Lw(i,j)$ be $0$, respectively.
For convenience, we define $\ND(\ell, 0) = \ND_0(\ell, 0) = \Lw(i, 0) = 0$ for any $\ell, i$, and $w$.
By the definition of $\ND$,
string $T[\ND(\ell, j)..j]$ is a minimal substring which contains a non-decreasing subsequence of length $\ell$
if and only if $\ND(\ell, j-1)< \ND(\ell, j)$ holds.
Similarly,
string $T[\ND_0(\ell, j)..j]$ is a minimal substring which contains a non-decreasing subsequence of length $\ell$ including at least one $0$'s
if and only if $\ND_0(\ell, j-1)<\ND_0(\ell, j)$ holds.
Let $q=\min\{i \in [m] \mid P[i..m] \text{ is a non-decreasing sequence}\}$.
We also define $\mathcal{M}$ as the set of all minimal CT-occurrence intervals for $P$ over $T$.
If $q=1$, i.e., $P[1..m]$ is a non-decreasing sequence,
by Observation~\ref{obs:nondecreasing},
$\mathcal{M}$ is the set of all minimal intervals $[\ell, r]$ such that $T[\ell.. r]$ contains a non-decreasing subsequence of length $m=|P|$.
If $q>1$, by Lemma~\ref{lem:bi_ct_match},
$\mathcal{M}$ is the set of all minimal intervals $[\ell, r]$
such that $T[\ell.. r]$ contains $w0^i1^{m-|w|-i}$ for some $i \ge 1$ as a subsequence where $w=P[1..q-1]$.
Note that the longest non-decreasing suffix $P[q.. m]$ of $P$ starts with character $0$ if $q > 1$.
Thus, the following observation holds for $\mathcal{M}$~(see Fig.~\ref{fig:ctmseq} for the case of $q > 1$):
\begin{observation}\label{obs:M}
  If $q=1$, $\mathcal{M}=\{[\ell, r] \mid \ell=\ND(m, r), r\in R\}$ holds, 
  where $R=\{j \in [n] \mid \ND(m, j-1)< \ND(m, j)\}$. 
  If $q>1$, let $\mathcal{M}'=\{[\ell, r] \mid \ell=\Lw(|w|, \ND_0(m-|w|, r)-1), r\in [n]\}$ where $w=P[1..q-1]$.
  Then, $\mathcal{M}=\{[\ell, r] \in \mathcal{M}' \mid \text{there is no interval } [\ell', r']\in \mathcal{M}' \text{ such that }[\ell', r'] \subsetneq [\ell, r]\}$ holds.
\end{observation}
\begin{figure}[hb]
  \centering
  \includegraphics[keepaspectratio, scale=0.4]{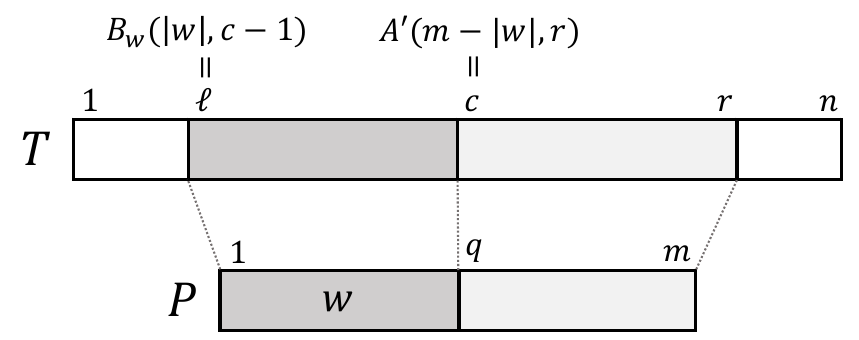}
  \caption{Illustration for an element $[\ell, r]$ of $\mathcal{M}'$.
    By Lemma~\ref{lem:bi_ct_match}, if $P$ is a CT-subsequence of a substring of $T$ ending at some position $r \in [n]$, then
    $P[q.. m]$, a non-decreasing suffix of $P$ including $0$, is a CT-subsequence of $T[c.. r]$ where $c = ND_0(m-q+1, r)$ and 
    $P[1.. q-1]$                                              is a    subsequence of $T[\Lw(|w|, c-1).. c-1]$ where $w = P[1.. q-1]$.
    Note that $[\ell, r]$ is not necessarily a minimal CT-occurrence interval for $P$ over $T$.
  }
  \label{fig:ctmseq}
\end{figure}

Now we describe an algorithm based on the above discussion which runs in $O(nm)$ time with $O(n)$ space.

\paragraph*{High-level description}
First, we compute $q=\min\{i \in [m] \mid P[i..m] \text{ is a non-decreasing sequence}\}$ in $O(m)$ time.
If $q=1$,
we compute the values of $\ND(m, j)$ for all integers $j \in [n]$.
By Observation~\ref{obs:M}, $\mathcal{M}$ can be computed in $O(n)$ time from the values of $\ND(m, \cdot)$.

If $q>1$, let $w=P[1..q-1]$.
We compute the values of $\ND_0(m-|w|, j)$ and $\Lw(|w|, j)$ for all integers $j \in [n]$.
By Observation~\ref{obs:M}, $\mathcal{M}'$ can be computed in $O(n)$ time from the values of $ND_0(m-|w|, \cdot)$ and $\Lw(|w|, \cdot)$.
Then, we can obtain $\mathcal{M}$ by removing non-minimal intervals from $\mathcal{M}'$.
Since every interval in $\mathcal{M}'$ is a sub-interval of $[1, n]$,
we can remove such non-minimal intervals in $O(n)$ time by using bucket sort.

Therefore, in each of the cases $q=1$ and $q>1$,
the remaining task is to compute the values of $\ND(m, j)$, $\ND_0(m-|w|, j)$ and $\Lw(|w|, j)$ for all $j \in [n]$.

\paragraph*{Computing \texorpdfstring{$\ND$}{ND}}
Let $\mathit{prev}_0(j)=\max(\{i < j \mid T[i]=0\}\cup\{0\})$.
The values of $\mathit{prev}_0(j)$ for all $j \in [n]$ can be computed in $O(n)$ time by scanning $T$ from right to left.
Let $Z(\ell, j)$ be the staring position of the longest suffix of $T[1.. j]$ that contains $0^\ell$ as a subsequence.
If no such suffix exists, let $Z(\ell, j)=0$.
The values of $Z$ can be computed in $O(nm)$ time by using the following recurrence:
\begin{align*}
  Z(\ell, j)&= \left\{ 
    \begin{alignedat}{2}   
        &Z(\ell-1, \mathit{prev}_0(j))&\quad&\text{if } T[j] = 0, \\   
        &Z(\ell, j-1)                 &\quad&\text{if } T[j] = 1.
    \end{alignedat} 
  \right.
\end{align*}

Then, we show how to compute $\ND$ (along with $Z$).
We first initialize $\ND(\ell, 0)=0$ for all $\ell \in [m]$ and
$\ND(1, j)=j$ for all $j \in [n]$.
Then, we can compute $\ND(m, j)$ for incremental $\ell = 1, 2, \ldots, m$ and $j = 1, 2, \ldots, n$
by using the following recurrence:
\begin{align*}
  \ND(\ell, j)&= \left\{ 
    \begin{alignedat}{2}   
        &\max\{Z(\ell, j), \ND(\ell, j-1)\}      &\quad&\text{if } T[j] = 0, \\   
        &\ND(\ell-1, j-1)                        &\quad&\text{if } T[j] = 1.
    \end{alignedat} 
  \right.
\end{align*}

We describe the above recurrence in detail in below.
Let $k = \ND(\ell, j) > 0$.
By the definition of $\ND$, there exists a subscript sequence $I(k, j_1, \ldots , j_{\ell-1})$ satisfying $k<j_1<\cdots<j_{\ell-1}\le j$
such that $T[k..j]$ contains a non-decreasing subsequence $T_I$.
If $T[j]=0$ and $j_{\ell-1}=j$, then $T[j_{\ell-1}]=0$ holds, and thus, $T_I=0^{\ell}$.
Hence, $k=Z(\ell, j)$ holds.
If $T[j] = 0$ and $j_{\ell-1}<j$, then $T[k..j-1]$ also contains $T_I$ as a subsequence.
Thus, $k=\ND(\ell, j-1)$ holds.
Finally, if $T[j]=1$, then $j_{\ell-1}=j$ always holds, and thus, $T[k..j-1]$ contains a non-decreasing subsequence of length $\ell-1$.
This leads to $k=\ND(\ell-1, j-1)$.
Therefore, the above recurrence for $\ND$ holds.

The time complexity of computing values of $\ND$ is $O(nm)$.
Although $\ND$ has $\Theta(nm)$ values,
the working space can be reduced to $O(n)$
since we only refer to the $\ell$th or the $(\ell-1)$th rows of $\ND$ and $Z$
when computing the $\ell$th rows of $\ND$ and $Z$, respectively.

\paragraph*{Computing \texorpdfstring{$\ND_0$}{ND_0}}
By the definitions of $\ND$ and $\ND_0$, The following relation holds. \begin{align*}
  \ND_0(\ell, j)= \left\{ 
    \begin{alignedat}{2}
          &\ND(\ell, j)      &\quad&\text{if } T[\ND(\ell, j)]=0,\\    
          &\mathit{prev}_0(\ND(\ell, j)) &\quad&\text{\rm{otherwise}}.   
    \end{alignedat} 
  \right.
\end{align*}
This relation implies that, for every $j \in [n]$,
if we already have $\ND(m-|w|, j)$, then $\ND_0(m-|w|, j)$ can be obtained $O(1)$ time.
Thus, we can compute the values of $\ND_0(m-|w|, j)$ for all $j \in [n]$ in $O(nm)$ time using $O(n)$ space.

\paragraph*{Computing \texorpdfstring{$\Lw$}{Lw}}
We note that the definition of $\Lw$ is equivalent to
the definition of the table, named $S$, used in~\cite{DasFGGK97} for the episode matching problem.
Thus, we can use the same recurrence as in~\cite{DasFGGK97} to compute $\Lw$:
\begin{align*}
  \Lw(i, j)= \left\{ 
    \begin{alignedat}{2}   
          &\Lw(i-1, j-1)                       &\quad&\text{if } w[i]=T[j], \\   
          &\Lw(i, j-1)                         &\quad&\text{\rm{otherwise}},
    \end{alignedat} 
  \right.
\end{align*}
with the initialization $\Lw(i, 0)=0$ for each $i \in [|w|]$ and $\Lw(0, j)=j+1$ for each $0 \leq j \leq n$.
We can obtain $\Lw(|w|, j)$ for all $j \in [n]$ in $O(n|w|) \subseteq  O(nm)$ time using $O(n)$ space~\cite{DasFGGK97}.

From the above discussion, we can compute the values of $\ND(m, j)$, $\ND_0(m-|w|, j)$ and $\Lw(|w|, j)$ for all $j \in [n]$ in $O(nm)$ time and $O(n)$ space.
Therefore, we obtain the following theorem:
\begin{theorem}
  \label{thm:ctmseq_bi_simple}
  {\CTMSeqprob} on binary strings
  can be solved in $O(nm)$ time and $O(n)$ space.
\end{theorem}

 \subsection{Conditional lower bound for {\CTMSeqprob}} \label{sec:lb_ctmseq}
In this section, we prove the following theorem. 
\begin{theorem}\label{th:ov2ctm}
  If there is an $\epsilon > 0$ such that {\CTMSeqprob} for an alphabet of size $4$
  can be solved in $O((|T||P|)^{1-\epsilon})$ time, then OVH and SETH are false.
\end{theorem}

\yynote*{added}{We show the proof of Theorem~\ref{th:ov2ctm} in the end of this section.
}To prove Theorem~\ref{th:ov2ctm},
we reduce the OV problem to the \emph{shortest-CT-MSeq} problem~({\SCTMSeqprob}) defined as follows, which is a relaxed version of {\CTMSeqprob}.

\begin{problem}[Shortest Cartesian-Tree Subsequence Matching problem ({\SCTMSeqprob})] \label{prob:s-stmseq}
  Given two strings $T$ and $P$, 
  find a shortest substring in $T$ such that $P$ is a CT-subsequence of the substring.
\end{problem}

We can easily reduce {\SCTMSeqprob} to {\CTMSeqprob}. 
Thus, we only consider the reduction from the OV problem to {\SCTMSeqprob} in the following proof.
Let $A = \{\alpha_1, \ldots, \alpha_{n_A}\}$ and $B = \{\beta_1, \ldots, \beta_{n_B}\}$
be the input of the OV problem where $n_B \le n_A$ and $\alpha_i, \beta_j \in \{0, 1\}^d$ for every $1\le i\le n_A, 1\le j\le n_B$.
From $A$ and $B$, we construct an input of {\CTMSeqprob}:
text $T_A \in \Sigma^*$ of length $20n_Ad + 8n_A + 1\in O(n_Ad)$ and
pattern $P_B \in \Sigma^*$ of length $4n_Bd + 2n_B + 1 \in O(n_Bd)$
where $\Sigma = \{\mt{0}, \mt{1}, \mt{2}, \mt{3}\}$.
\tmnote*{changed}{While our reduction is similar to the reduction from the OV problems to the \emph{episode matching problem}~\cite{BilleGMSW22},
our construction for text and pattern is non-trivial and includes new ideas that exploit properties of Cartesian trees.
}The key idea of our reduction is to construct pattern $P_B$
  so that the Cartesian tree of $P_B$ has \emph{comb-like} structures.
  To achieve this, we insert characters $c$ between gadget strings in $P_B$
  such that $c$ is smaller than any other character appearing in the gadgets~(see also Fig.~\ref{fig:v2_beta}).
Thanks to such comb-like structures,
  we can prevent interference between Cartesian trees of gadget strings.

First, we define the \emph{coordinate gadgets} $C_1$ and $C_2$:
\begin{eqnarray*}
  C_1(a) &=& 
  \begin{cases}
    \mt{2323} & \mathrm{if} \; a = 0, \\
    \mt{3232} & \mathrm{if} \; a = 1.
  \end{cases} \\
  C_2(b) &=& 
  \begin{cases}
    \mt{323} & \mathrm{if} \; b = 0, \\
    \mt{222} & \mathrm{if} \; b = 1.
  \end{cases} 
\end{eqnarray*}
It can be seen that both $C_2(0) = \mt{323}$ and $C_2(1) = \mt{222}$ are CT-subsequences of $C_1(0) = \mt{2323}$. 
  Also, $C_2(0) = \mt{323}$ is a CT-subsequence of $C_1(1) = \mt{3232}$ 
  but $C_2(1) = \mt{222}$ is not.
  Thus, we have the following observation:
  \begin{observation}\label{obs:cgadget}
    String $C_2(b)$ is a CT-subsequence of string $C_1(a)$ if and only if $a\cdot b = 0$.
  \end{observation}

Next, we construct the \emph{vector gadgets} $V_1(\alpha)$ and $V_2(\beta)$ for all $\alpha\in A$ and $\beta\in B$:
\begin{alignat*}{12}
  V_1(\alpha) &={}& &\mt{1} \;C_1(\alpha[1])\;&\mt{1}&\;C_1(\alpha[2])\;&\mt{1}&\;\cdots\;&\mt{1}&\; C_1(\alpha[d])\; &\mt{1}& \\
  V_2(\beta)  &={}& &\mt{1} \;C_2(\beta[1])\;&\mt{1}&\;C_2(\beta[2])\;&\mt{1}&\;\cdots\;&\mt{1}&\; C_2(\beta[d])\; &\mt{1}&
\end{alignat*}
Note that $|V_1(\alpha)| = 5d+1$ and $|V_2(\beta)| = 4d+1$.
Fig.~\ref{fig:v2_beta} shows an example of $V_2(\beta)$ and $CT(V_2(\beta))$ where $\beta = (0, 1, 0, 0, 1)$.
\begin{figure}[tb]
  \centering
  \includegraphics[keepaspectratio, scale=0.45]{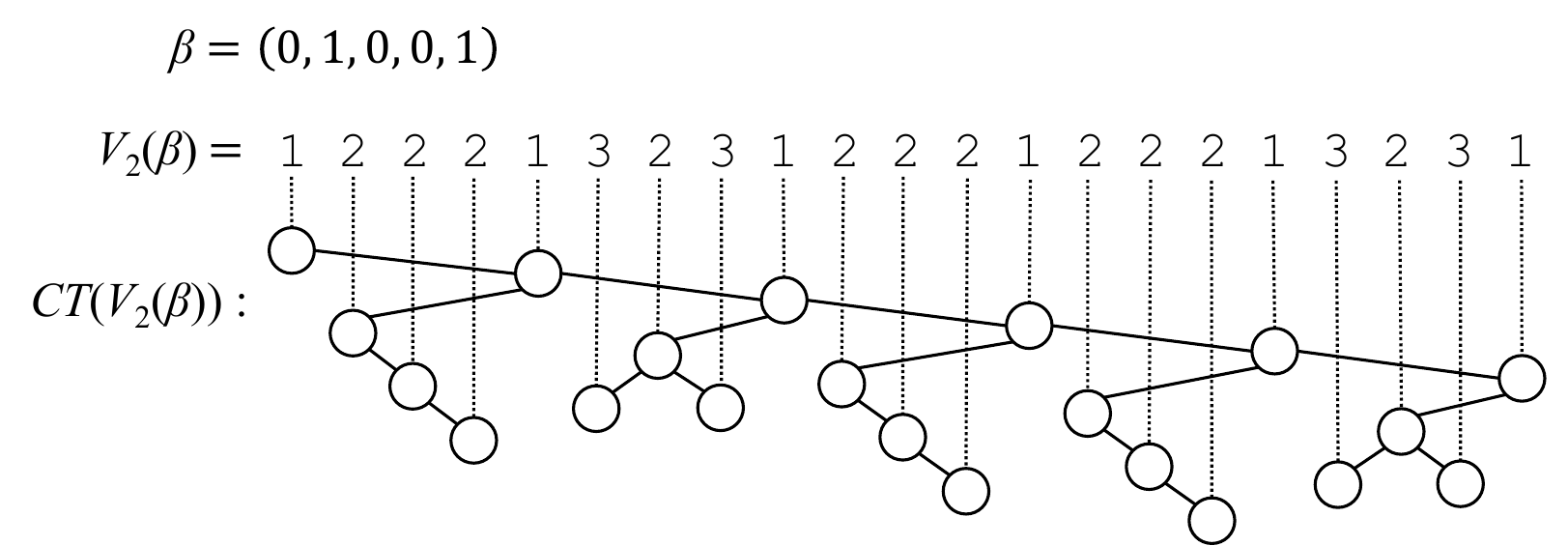}
  \caption{
    Illustration for $V_2(\beta)$ and $CT(V_2(\beta))$ where $\beta = (0, 1, 0, 0, 1)$.
    Since coordinate gadgets consist only of $\mt{2}$ and $\mt{3}$,
    each coordinate gadget is the left subtree of some $\mt{1}$, resulting in an unbalanced Cartesian tree.
    We call such a structure comb-like.
    Note that $V_2(\beta)$ contains six $\mt{1}$'s since $d = 5$ in this example.}
  \label{fig:v2_beta}
\end{figure}
The following Lemma~\ref{lem:vg} holds.
\begin{lemma} \label{lem:vg}
  String $V_2(\beta)$ is a CT-subsequence of string $V_1(\alpha)$ if and only if $\alpha$ and $\beta$ are orthogonal.
\end{lemma}
\begin{proof}
  For each $1 \le i \le d$, Cartesian trees of $C_1(\alpha[i])$ and $C_2(\beta[i])$ belong to the left subtree of the $(i+1)$th character $\mt{1}$ 
  in $CT(V_1(\alpha))$ and $CT(V_2(\beta))$, respectively,
  since $C_1(\alpha[i])$ and $C_2(\beta[i])$ consist only of $\mt{2}$ and $\mt{3}$.
Now we consider the following two cases:
  \begin{enumerate}
    \item If $\alpha$ and $\beta$ are orthogonal, 
      $C_2(\beta[i])$ is a CT-subsequence of $C_1(\alpha[i])$ for all $1 \leq i \leq d$
      by Observation~\ref{obs:cgadget}.
      Therefore, by assigning every $\mt{1}$ in $V_2(\beta)$ to every $\mt{1}$ in $V_1(\alpha)$, and
      $C_2(\beta[i])$ to $C_1(\alpha[i])$ for each $1 \le i \le d$,
      we obtain a subsequence of $V_1(\alpha)$ that CT-matches $V_2(\beta)$.
    \item Otherwise, there exists $k$ with $\alpha[k] = \beta[k] = 1$.
      By Observation~\ref{obs:cgadget},
      $C_2(\beta[k])=C_2(1)=\mt{222}$ is not a CT-subsequence of $C_1(\alpha[k]) = C_1(1) = \mt{3323}$.

      For a Cartesian tree of a string, the path from the root to the rightmost node is called the \emph{frontier} of the tree.
      From the form of $V_2(\beta)$,
      the frontier of $CT(V_2(\beta))$ consists of $(d+1)$ $\mt{1}$'s 
      and their left subtrees are of size three~(see Fig.~\ref{fig:v2_beta} for examples).
      We now prove Claim~\ref{claim}.
      \begin{claim}\label{claim}
        Assume that there is a \tmnote*{fixed}{subsequence} $S$ of $V_1(\alpha)$ that CT-matches $V_2(\beta)$.
        The frontier of $CT(S)$ consists only of character $\mt{1}$.
      \end{claim}
      \begin{proof}
        Assume on the contrary that the frontier of $CT(S)$ contains a character $c \ne \mt{1}$.
        By the definition of Cartesian trees, the left subtree of $c$ consists only of characters that are greater than $c$.
        Thus, $c = \mt{3}$ is invalid, so $c = \mt{2}$.
        Further, the left subtree of $\mt{2}$ is $CT(\mt{333})$ since the size of the left subtree is three.
        Since every coordinate gadget has at most two $\mt{3}$'s, selecting $\mt{333}$ from $V_1(\alpha)$ results in
        at least one $\mt{1}$ being \emph{unavailable} as an element of $S$.
        Also, the rightmost (lowest) frontier node of $CT(S)$ is $\mt{2}$, and thus,
        the rightmost $\mt{1}$ in $V_1(\alpha)$ cannot be contained in $S$.
        Hence, there are at least $t+1$ unavailable $\mt{1}$'s
        where $t$ is the number of $\mt{2}$'s in the frontier of $CT(S)$.
        Therefore, the number of available $\mt{1}$'s is at most $d-t$ and
        the number of $1$'s in the frontier of $CT(S)$ is $d+1-t$, a contradiction.
      \end{proof}
      By Claim~\ref{claim}, each $\mt{1}$ in $V_2(\beta)$ must be aligned to each $\mt{1}$ in $V_1(\alpha)$.
      Thus, $C_2(\beta[i])$ should be assigned to $C_1(\alpha[i])$ for each $1 \le i \le d$.
      However, as mentioned before, $C_2(\beta[k])$ is not a CT-subsequence of $C_1(\alpha[k])$.
      Therefore, there is no \tmnote*{fixed}{subsequence} $S$ of $V_1(\alpha)$ that CT-matches $V_2(\beta)$.
\end{enumerate}
\end{proof}

Now we construct text $T_A$ and 
pattern $P_B$,
joining dummy vector gadget $V_1(\zero) = \mt{1}(C_1(0)\mt{1})^d$
where $\zero$ denotes the $d$-dimensional zero vector:
\begin{eqnarray*}
  T_A &={}& 
  \mt{0} \left(\prod^{n_A}_{i=1}\left(V_1(\alpha_i) \; \mt{0} \; V_1(\zero) \; \mt{0}\right)\right)^2 \\
  P_B &={}& \mt{0} \; V_2(\beta_1) \; \mt{0} \; V_2(\beta_2) \; \mt{0} \; \cdots \; \mt{0} \; V_2(\beta_{n_B}) \; \mt{0}
\end{eqnarray*}

Note that, for any $\beta \in \{0, 1\}^d$, $V_2(\beta)$ is a CT-subsequence of $V_1(\zero)$.
We show the following lemmas:
\begin{lemma}\label{lem:1_2}
  If some $\alpha \in A$ and $\beta \in B$ are orthogonal, then $T_A$ has a substring $S$
  such that $|S| \le (n_B-1)(10d+4) +1$ and $P_B$ is a CT-subsequence of $S$. 
\end{lemma}
\begin{proof}
  We construct an alignment of (substrings of) $P_B$ to $T_A$ such that
  the alignment implies a subsequence of $T_A$ that CT-matches $P_B$.
  For each $i$ with $1 \le i \le n_B$,
  $CT(V_2(\beta_i))$ is the left subtree of the $(i+1)$th $\mt{0}$ in $CT(P_B)$
  since $V_2(\beta_i)$ contains only characters $\mt{1}$, $\mt{2}$ and $\mt{3}$.
Namely, similar to the case of $CT(V_2(\beta))$,
  $CT(P_B)$ has a comb-like structure with length-$(n_B+1)$ frontier consisting only of $\mt{0}$'s.
  Thus, by assigning every $\mt{0}$ in $P_B$ to some $\mt{0}$ in $T_A$,
  we can focus only on isomorphism between Cartesian trees of vector gadgets.

  Let $\alpha_i \in A$ and $\beta_j \in B$ be vectors which are orthogonal. 
  If $j \leq i$, we align $V_2(\beta_j)$ to the \emph{first} occurrence of $V_1(\alpha_i)$ in $T_A$ by Lemma~\ref{lem:vg}.
  By doing so, it is guaranteed that
  every vector gadget $V_2(\beta_s)$ in $P_B$ with $s \ne j$ can be aligned to
  some vector gadget $V_1(\zero)$ in $T_A$. 
  Symmetrically,
  if $j > i$, we align $V_2(\beta_j)$ to the \emph{second} second occurrence $V_1(\alpha_i)$ in $T_A$.

  Here, we consider three the following cases w.r.t.~the value of $j$.
\begin{enumerate}
    \item If $j=1$, all the vector gadgets in $P_B$, except $V_2(\beta_1)$, 
      can be aligned to $V_1(\zero)$'s to the right of $V_1(\alpha_i)$
      in order from closest to $V_1(\alpha_i)$, 
      and all $\mt{0}$'s in $P_B$ are aligned to $\mt{0}$'s which interleave vector gadgets in $T_A$.
      Then, the length of such substring of $T_A$ is $(n_B-1)(10d+4) +1$ 
      since
      (1) $\mt{0}  V_2(\beta_1)  \mt{0}  V_2(\beta_2)  \mt{0}$ is a CT-subsequence of
      the substring $\mt{0}  V_1(\alpha_i)  \mt{0}  V_1(\zero)  \mt{0}$ of $T_A$
      of length $10d + 5$, and
(2) $\prod_{k=3}^{n_B}(V_2(\beta_k)  \mt{0})$ is a CT-subsequence of 
      the substring $\prod_{k = i+1}^{i+n_B-2}(V_1(\alpha_k)\mt{0}V_1(\zero)\mt{0})$ of $T_A$
      of length $(n_B-2)(10d + 4)$.
\item If $j=n_B$, symmetrical to the case of $j=1$, we can align $P_B$ to $T_A$ appropriately.
      Then the length of such substring of $T_A$ is $(n_B-1)(10d+4) +1$
      since
      (1) $V_2(\beta_n)  \mt{0}$ is a CT-subsequence of 
      the substring $V_1(\alpha_i)  \mt{0}$ of $T_A$ of length $5d + 2$,
      (2) $\mt{0}  V_2(\beta_1)  \mt{0}$ is a CT-subsequence of
      the substring $\mt{0}  V_1(\zero)  \mt{0}$ of $T_A$ of length $5d + 3$, and
      (3) $\prod_{k=2}^{n_B-1}(V_2(\beta_k)  \mt{0})$ is a CT-subsequence of
      $\prod_{k=i-n_B+2}^{i-1}(V_1(\alpha_k)\mt{0}V_1(\zero)\mt{0})$ of $T_A$
      of length $(n_B-2)(10d + 4)$.
    \item Otherwise, both $V_2(\beta_{j-1})$ and $V_2(\beta_{j+1})$ can be aligned to $V_1(\zero)$ adjacent to $\mt{0}V_1(\alpha_i)\mt{0}$, 
      and all the remaining parts of $P_B$ are aligned in the same way as the previous cases (see also Fig.~\ref{fig:alignment} for illustration).
      Here, we discuss only the case $j \le i$.
      The other case $i < j$ can be shown in a symmetric way.
The length of the resulting substring of $T_A$ is $(n_B-1)(10d+4) + 1 - (5d+2)$ since
      (1) $V_2(\beta_j)  \mt{0}  V_2(\beta_{j+1})  \mt{0}$ is a CT-subsequence of
      the substring $V_1(\alpha_i)  \mt{0}  V_1(\zero)  \mt{0}$ of $T_A$ of length $10d + 4$,
      (2) $\prod_{k=j+2}^{n_B}(V_2(\beta_{k})  \mt{0})$ is a CT-subsequence of 
      the substring $\prod_{k=i+1}^{i-j+n_B-1}(V_1(\alpha_k)\mt{0}V_1(\zero)\mt{0})$ of $T_A$ of length $(n_B-j-1)(10d + 4)$,
      (3) $\prod_{k=2}^{j-1}(V_2(\beta_k)  \mt{0})$ is a CT-subsequence of
      the substring $\prod_{k=i-j+2}^{i-1}(V_1(\alpha_k)\mt{0}V_1(\zero)\mt{0})$ of $T_A$ of length $(j-2)(10d + 4)$, and
      (4) $\mt{0}  V_2(\beta_1)  \mt{0}$ is a CT-subsequence of 
      the substring $\mt{0}  V_1(\alpha_{i-j+1})  \mt{0}$ is $5d + 3$ of $T_A$ of length $5d + 3$.
      Thus, the length of the concatenation of the four consecutive substrings of $T_A$ is shorter than $(n_B-1)(10d+4)+1$.
  \end{enumerate}
\end{proof}
\newcommand{\asp}{~\;\:}
\newcommand{\bsp}{\hspace{0.55mm}}
\begin{figure}[tb]
  \centering
  \hspace{10.5mm}\includegraphics[scale=0.434]{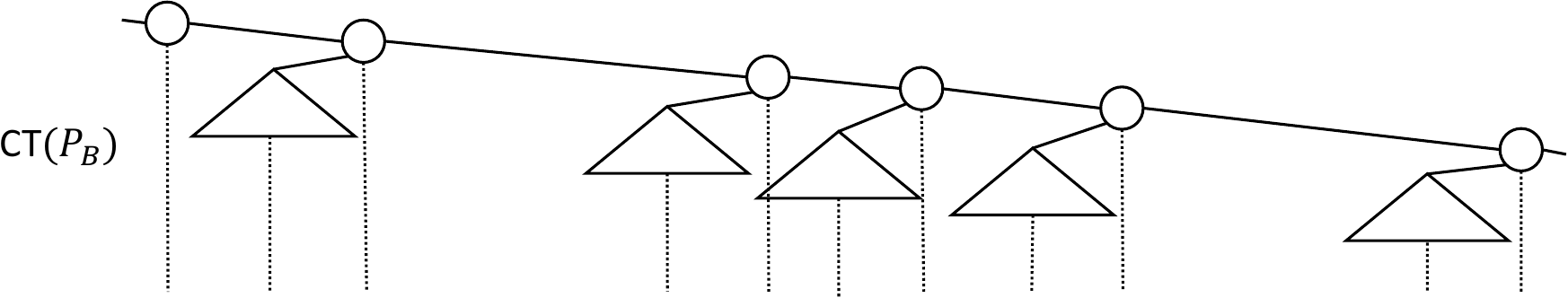}
  \vspace{-4.5mm}
  \begin{align*}
    T_A &=  \cdots \mt{0} \asp V_1(\zero) \asp \mt{0} \, V_1(\alpha_{i-1}) \, \mt{0} \asp V_1(\zero) \asp \mt{0} \, V_1(\alpha_i) \, \mt{0} \asp V_1(\zero) \asp \mt{0} \, V_1(\alpha_{i+1}) \, \mt{0} \asp V_1(\zero) \asp \mt{0}\cdots
    \\
    P_B &= \cdots \mt{0} \bsp V_2(\beta_{j-2}) \bsp \mt{0} \hspace{1.685cm} V_2(\beta_{j-1}) \bsp \mt{0} \, V_2(\beta_{j}) \, \mt{0} \bsp V_2(\beta_{j+1}) \bsp \mt{0} \hspace{1.685cm} V_2(\beta_{j+2}) \bsp \mt{0}\cdots
  \end{align*}
  \caption{
    Illustration for an alignment of pattern $P_B$ in text $T_A$ when
    $\alpha_i \in A$ and $\beta_j \in B$ are orthogonal for $i, j$ with $1 < j \leq i$ and $j < n_B$.
    The upper part shows a shape of a subtree of $CT(P_B)$.
    The subtree forms a comb-like structure because vector gadgets $V_2(\cdot) \in \{\mt{1}, \mt{2}, \mt{3}\}^*$ are interleaved with $\mt{0}$s.
    The lower part shows a part of the alignment of $P_B$ in $T_A$
    described in the third case of the proof of Lemma~\ref{lem:1_2}.
    The alignment confirms that $P_B$ is a CT-subsequence of $T_A$.
  }
  \label{fig:alignment}
\end{figure}

\yynote*{added}{Next, we consider the case where $\alpha$ and $\beta$ are not orthogonal for all pairs.
To facilitate the proof in the case, we prove the following lemma.
}\begin{lemma}\label{lem:num_of_one}
  Assume that $\alpha$ and $\beta$ are not orthogonal for all pairs of $\alpha \in A$ and $\beta \in B$.
  If $T_A$ has a substring $S$ such that $P_B$ is a CT-subsequence of $S$, then
  the frontier of $CT(S)$ contains at most one $\mt{1}$.
\end{lemma}
\begin{proof}
  Let $S[p] = \mt{1}$ be a character which belongs to the frontier of $CT(S)$.
  By the definition of Cartesian trees, the left subtree $L$ of $S[p]$ consists only of $\mt{2}$ and $\mt{3}$.
  Since $L$ must be isomorphic to $CT(V_2(\beta))$ for some $\beta \in B$,
  the frontier of $L$ consists only of $\mt{2}$~(similar to Claim~\ref{claim}).
  Again, by the definition of Cartesian trees,
  the left subtrees of $\mt{2}$'s in $L$ consist only of $\mt{3}$, i.e., $CT(\mt{333})$.
  Thus, $L$ is isomorphic to $CT(V_2(\tilde{\beta}))$ where $\tilde{\beta} = (1, 1, \ldots, 1)$.
  Since there are at most one such vector $\tilde{\beta} \in B$,
  the frontier of $CT(S)$ can contain at most one $\mt{1}$.
\end{proof}

\yynote*{added}{Now, we start to prove the following lemma.
}\begin{lemma}\label{lem:2_2}
  If $\alpha$ and $\beta$ are not orthogonal for all pairs of $\alpha \in A$ and $\beta \in B$, then $T_A$ dose not have a substring $S$
  such that $|S| \le (n_B-1)(10d+4) +1$
and $P_B$ is a CT-subsequence of $S$. 
\end{lemma}
\begin{proof}
By Lemma~\ref{lem:num_of_one}, any substring of $T_A$ which CT-matches $P_B$
  contains at most one $\mt{1}$, in turn, it contains at least $n_B$ $\mt{0}$'s.
  That is, the first $n_B$ $\mt{0}$'s in $P_B$ should be aligned to some $\mt{0}$'s in $T_A$.
  Recall that, for any $\alpha \in A$ and $\beta\in B$,
  $V_2(\beta)$ is not a CT-subsequence of $V_2(\alpha)$, but
  it is a CT-subsequence of $V_1(\zero)$.
Thus, the shortest substring of $T_A$ which contains prefix $\left(\prod_{k=1}^{n_B-2}\mt{0}V_2(\beta_k)\right)\cdot\mt{0}V_2(\beta_{n_B-1})\mt{0}$ of $P_B$ as a CT-subsequence
  is $\left(\prod_{k = s}^{n_B-1+s}\mt{0}V_1(\zero)\mt{0}V_1(\alpha_k)\right)\cdot\mt{0}V_1(\zero)\mt{0}$
  for some integer $s \ge 1$.
The remaining part of $P_B$ is $V_2(\beta_{n_B})\mt{0}$.
  Since $V_2(\beta_{n_B})$ is not a CT-subsequence of $V_1(\alpha_{n_B+s})$,
  the shortest substring of $T_A$ that starts with $V_1(\alpha_{n_B+s})$ and contains $V_2(\beta_{n_B})\mt{0}$ as a CT-subsequence
  ends within $\mt{0}V_1(\zero)\mt{0} = \mt{01}(\mt{23231})^d\mt{0}$.
  Furthermore, the ending position within $\mt{01}(\mt{23231})^d$ of such substring is greater than two
  since the length-$2$ suffix of $P_B$ is $\mt{10}$, which does not CT-match $\mt{01}$.
Therefore, the length of a shortest substring of $T_A$ that contains $P_B$ as a CT-subsequence is at least
  $(n_B-2)(10d+4) + (5d+3) + 5d+3 = (n_B-1)(10d+4) + 2 > (n_B-1)(10d+4) + 1$.
\end{proof}

From Lemmas~\ref{lem:1_2} and~\ref{lem:2_2}, 
we have a reduction from the OV problem to {\SCTMSeqprob} (and thus, to {\CTMSeqprob}) such that
$T_A \in O(n_Ad)$, $P_B \in O(n_Bd)$. 
If {\CTMSeqprob} can be solved in $O((|T_A||P_B|)^{1-\epsilon})$ time, 
then the OV problem can be solved in $O((n_An_Bd^2)^{1-\epsilon}) \subseteq O((n_An_B)^{1-\epsilon}\mathsf{poly}(d))$ time.
When $n_A = n_B$, the result contradicts OVH/SETH,
and when $n_B = \Theta({n_A}^\lambda)$ for $0 < \lambda \le 1$ and $d \le n^{o(1)}$, 
the result contradicts UOVH/SETH~(cf.~Section~\ref{sec:hardness}).
Therefore, Theorem~\ref{th:ov2ctm} holds.
 \section{Longest common CT-subsequence problems} \label{sec:ctlcs}

\subsection{{\CTLCSprob} algorithms for general ordered alphabets}
\label{sec:alg_on_general}

In this section, we propose an algorithm for solving {\CTLCSprob} for general ordered alphabets.
An $O(n^6)$-time and $O(n^4)$-space algorithm, which is our main result will be given in Section~\ref{subsec:fast_alg}.
We start from explaining an $O(n^8)$-time and $O(n^6)$-space algorithm for simplicity (Section~\ref{subsec:slow_alg}).

For each $c \in \Sigma$, let $P_c(S) = \{i \mid S[i] = c\}$
and let $S'$ be the string of length $|S|$ such that $S'[i] = (S[i], r_i)$
for $1 \leq i \leq |S|$, where $r_i$ is the rank of $i$ in $P_{S[i]}(S)$.
For ordered pairs $(c, r)$ and $(c', r')$ of characters and integers,
let $(c, r) < (c', r')$ iff (1) $c < c'$ or (2) $c = c'$ and $r < r'$.
Then, it holds that $\CT(S)$ and $\CT(S')$ are isomorphic.
Thus, without loss of generality,
we can assume that all characters in the string $S$ are distinct.
The same assumption applies to the other string $T$,
but $S$ and $T$ may share the same characters.

\subsubsection{$O(n^8)$-time and $O(n^6)$-space algorithm}
\label{subsec:slow_alg}

We refer to a pair $(i, j) \in [n]^2$ of positions in $S$ and $T$ as a \emph{pivot}.
Our algorithm in the general case is based on 
the idea of \emph{pivoted Cartesian-trees} from Oizumi et al.~\cite{OizumiKMIA22} defined as follows.
\begin{figure}[tb]
    \centering
    \includegraphics[keepaspectratio, scale=0.4]{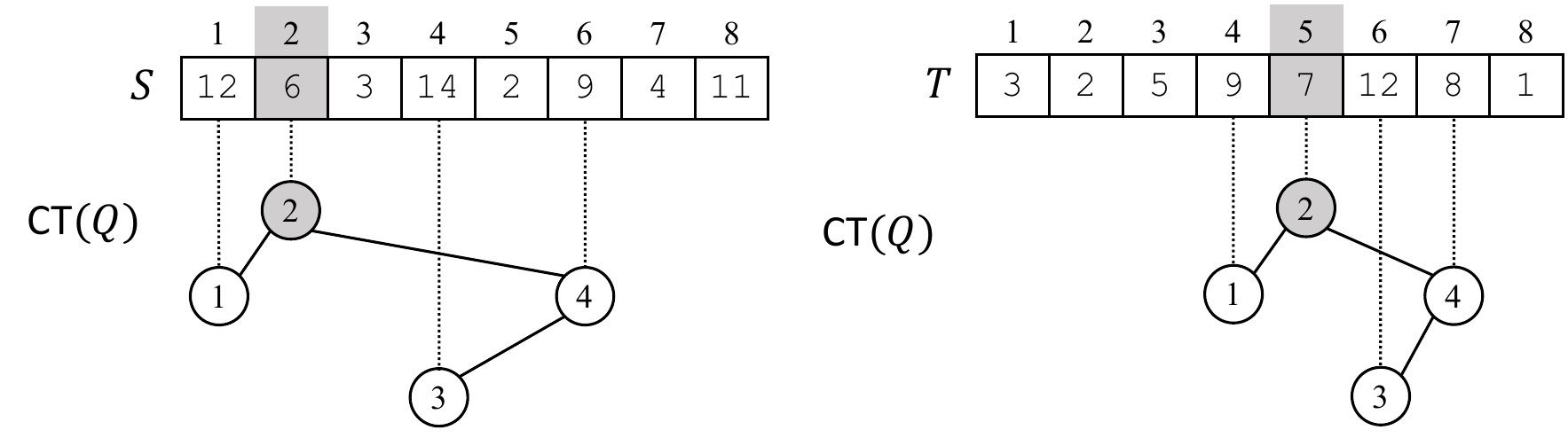}
    \caption{
        $Q=\str{\mt{12},\mt{6},\mt{14},\mt{9}}$ is a fixed longest common CT-subsequence of 
        $S=\str{\mt{12},\mt{6},\mt{3},\mt{14},\mt{2},\mt{9},\mt{4},\mt{11}}$ and $T=\str{\mt{3},\mt{2},\mt{5},\mt{9},\mt{7},\mt{12},\mt{8},\mt{1}}$
        with pivot $(2,5)$.
    }   
    \label{fig:fix_ctlcs}
\end{figure}

\begin{definition}[Fixed (longest) common CT-subsequence]
\label{def:fctlcs}
    Let $(i, j) \in [n]^2$ be a pivot of strings $S$ and $T$.
    A string $Q$ is said to be a \emph{fixed common CT-subsequence} (\emph{f-CT-CS}) of $S$ and $T$ with pivot $(i, j)$
    if there exist subscript sequences $I,J \in \mathcal{I}_{|Q|}^n$ such that
    \begin{itemize}
        \item $\CT(Q)=\CT(S_I)=\CT(T_J)$,
        \item $S[i]=\minst(S_I)$, and 
        \item $T[j]=\minst(T_J)$.
    \end{itemize}
    Moreover, a string $Q$ is said to be the \emph{fixed longest common CT-subsequence} (\emph{f-CT-LCS}) of $S$ and $T$ with pivot $(i, j)$
    if there are no f-CT-CS with pivot $(i, j)$ longer than $Q$ (see also Fig.~\ref{fig:fix_ctlcs}).
\end{definition}

Our solution is a dynamic programming based on the f-CT-LCS.
We also consider the f-CT-LCS for substrings of $S$ and $T$.
We will use positions $i, j \in [n]$ of the input strings to indicate a pivot for substrings,
namely, we say pivot $(i, j) \in [n]^2$ of substrings $S' = S[\ell_1..r_1]$ and $T' = T[\ell_2..r_2]$
instead of pivot $(i-\ell_1+1, j-\ell_2+1) \in [r_1-\ell_1+1] \times [r_2-\ell_2+1]$ of $S'$ and $T'$,
where $\ell_1 \leq i \leq r_1$, $\ell_2 \leq j \leq r_2$.
Let $C(i, j, \ell_1, r_1, \ell_2, r_2)$ be the length of 
the f-CT-LCS of substrings $S[\ell_1..r_1]$ and $T[\ell_2..r_2]$ with pivot $(i, j)$.
It is clear from the definition that 
\[
    \ctlcs(S,T)=\max\{C(i, j, 1, n, 1,n) \;|\; 1 \leq i \leq n, 1 \leq j \leq n\}
\]
holds.
The following lemma shows the main idea of computing $C(i, j, 1, n, 1, n)$ by a dynamic programming
(see also Fig.~\ref{fig:lem} for an illustration).
\begin{figure}[tb]
    \centering
    \includegraphics[keepaspectratio, scale=0.4]{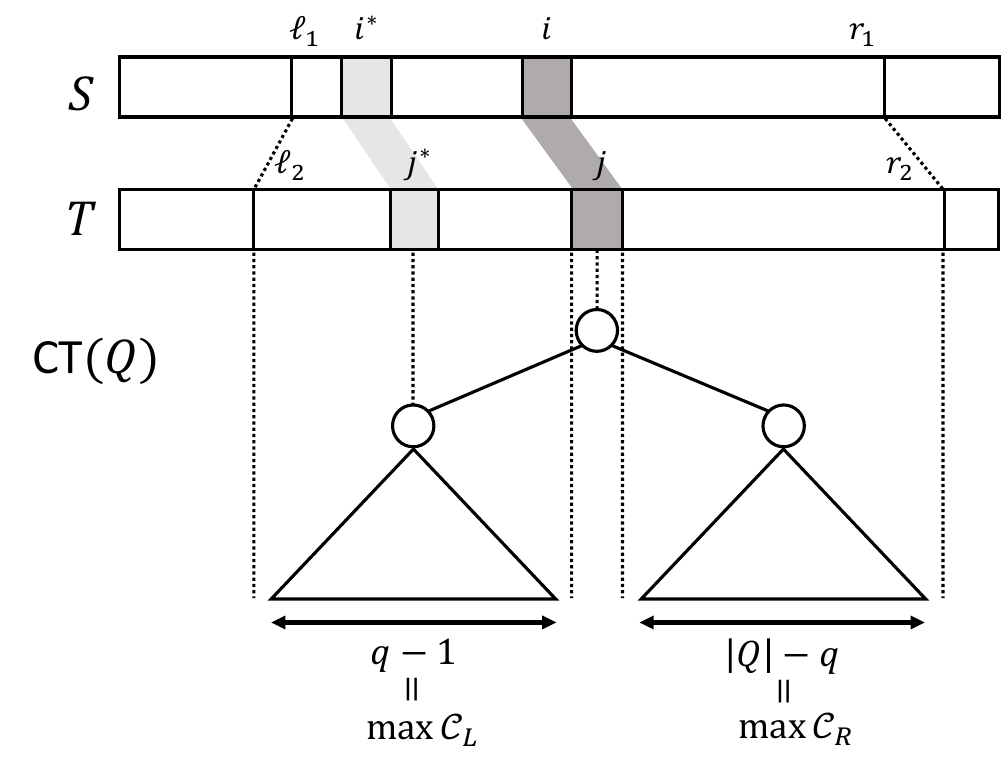}
    \caption{Sketch of our recurrence by Lemma~\ref{lem:main_idea}.}
    \label{fig:lem}
\end{figure}

\begin{lemma}
  \label{lem:main_idea}
  For any $(i, j, \ell_1, r_1, \ell_2, r_2) \in [n]^6$ that satisfies $\ell_1 \leq i \leq r_1$, $\ell_2 \leq j \leq r_2$,
  define $\mathcal{C}_L$ and $\mathcal{C}_R$ as follows:
  \begin{linenomath}
    \begin{align*}
      \mathcal{C}_L = \{&C(i', j', \ell_1, i-1, \ell_2, j-1) \\ &\mid S[i']>S[i], T[j']>T[j], \ell_1 \leq i' \leq i-1, \ell_2 \leq j' \leq j-1\} \cup \{0\}, \\
      \mathcal{C}_R = \{&C(i', j', i+1, r_1, j+1, r_2) \\ &\mid S[i']>S[i], T[j']>T[j], i+1 \leq i' \leq r_1, j+1 \leq j' \leq r_2\} \cup \{0\}.
    \end{align*}
  \end{linenomath}Then the recurrence 
    $C(i, j, \ell_1, r_1, \ell_2, r_2) = \max \mathcal{C}_L+\max \mathcal{C}_R+1$ holds.
\end{lemma}

\begin{proof}
    Let $Q$ be the f-CT-LCS of $S[\ell_1..r_1]$ and $T[\ell_2..r_2]$ with $(i, j)$.
    By Definition~\ref{def:fctlcs},
    there exist subscript sequences $I=(i_1,\ldots,i_{|Q|})$ and $J=(j_1,\ldots,j_{|Q|})$ that satisfy $\CT(Q)=\CT(S_I)=\CT(T_J)$.
    It is also clear from the definition that $S[i]=\minst(S_I)$ and $T[j]=\minst(T_J)$ hold.
    Let $q = \minid(Q)$.
    We show that $q-1 = \max \mathcal{C}_L$ holds.
    \begin{itemize}
        \item Assume that $q=1$.
        In this case, we need to show $\mathcal{C}_L = \{0\}$.
        Suppose on the contrary that there exists $(i', j')$ 
        such that $S[i']>S[i]$,  $T[j']>T[j]$, $\ell_1 \leq i' \leq i-1$, and $\ell_2 \leq j' \leq j-1$.
        Let $I^* = (i', i_1, \ldots, i_{|Q|})$, $J^* = (j', j_1, \ldots, j_{|Q|})$, 
        and $Q^* = \alpha \cdot Q$ where $\alpha$ is a character in $\Sigma$ that satisfies $\alpha > Q[q]$.
Also, $\CT(Q^*)=\CT(S_{I^*})=\CT(T_{J^*})$ holds
        since $\CT(Q)=\CT(S_I)=\CT(T_J)$, $\alpha > Q[q]$, $S[i']>S[i]$, and $T[j']>T[j]$ hold.
Moreover, $Q^*$ is an f-CT-CS of $S[\ell_1..r_1]$ and $T[\ell_2..r_2]$ with $(i, j)$,
        since $S[i] = \minst(S_{I^*})$ and $T[j] = \minst(T_{J^*})$.
        This contradicts the fact that $Q$ is the f-CT-LCS of $S[\ell_1..r_1]$ and $T[\ell_2..r_2]$ with $(i, j)$ (from $|Q^*|>|Q|$).
        Thus, $\mathcal{C}_L = \{0\}$ and $\max \mathcal{C}_L = 0$.
        \item Assume that $q>1$.
        Let $p_1$ (resp., $p_2$) be the predecessor of $i$ in $I$ (resp., the predecessor of $j$ in $J$).
        By the definition of the Cartesian-tree, 
        $\CT(Q[1..q-1]) = \CT(S_{I[i_1:p_1]}) = \CT(T_{J[j_1:p_2]})$ holds,
        and there exist $i^*$ and $j^*$ such that $S[i^*] = \minst(S_{I[i_1:p_1]}) (> S[i])$ and $T[j^*]=\minst(T_{J[j_1:p_2]}) (> T[j])$.
        This implies that $Q[1..q-1]$ is an f-CT-CS of $S[\ell_1..i-1]$ and $T[\ell_2..j-1]$ with $(i^*, j^*)$.
        Thus, $|Q[1..q-1]| = q-1 \leq \max \mathcal{C}_L$ holds.
        In the rest of this case, we show $q-1 \geq \max \mathcal{C}_L$ to prove the equality.
        Suppose on the contrary that $q-1 < \max \mathcal{C}_L$.
        Since $0 < q-1 < \max \mathcal{C}_L$ (from assumptions),
        there exists $(i'', j'') \in [n]^2$ such that $C(i'', j'', \ell_1, i-1, \ell_2, j-1) = \max \mathcal{C}_L$.
        Let $Q''$ be an f-CT-LCS of $S[\ell_1..i-1]$ and $T[\ell_2..j-1]$ with $(i'', j'')$.
        Then there exist subscript sequences $I''$ and $J''$ over $\{\ell_1, \ldots, i-1\}$ and $\{\ell_2, \ldots, j-1\}$, respectively,
        such that $\CT(Q'') = \CT(S_{I''}) = \CT(T_{J''})$.
        Let $\hat{I}$ denote the subscript sequence that is the concatenation of $I''$ and $I[i:i_{|Q|}]$,
        and $\hat{J}$ denote the subscript sequence that is the concatenation of $J''$ and $J[j:j_{|Q|}]$.
        Then $\CT(Q'' \cdot Q[q..|Q|]) = \CT(S_{\hat{I}}) = \CT(T_{\hat{J}})$,
        $S[i] = \minst(S_{\hat{I}})$, and $T[j] = \minst(T_{\hat{J}})$ hold.
        This implies that $Q'' \cdot Q[q..|Q|]$ is an f-CT-LCS of $S[\ell_1..i-1]$ and $T[\ell_2..j-1]$ with $(i, j)$.
        However, $|Q'' \cdot Q[q..|Q|]| = |Q''|+|Q[q..|Q|]| = \max \mathcal{C}_L+|Q|-q+1 > q-1+|Q|-q+1 = |Q|$ holds.
        This contradicts to the fact that $Q$ is the f-CT-LCS of $S[\ell_1..i-1]$ and $T[\ell_2..j-1]$ with $(i, j)$.
        Thus $q-1 = \max \mathcal{C}_L$ also holds for $q>1$.
    \end{itemize}
    We can also prove $|Q|-q = \max \mathcal{C}_R$ by a symmetric manner.
    Therefore, $|Q| = q-1 + |Q|-q + 1 = \max \mathcal{C}_L+\max \mathcal{C}_R+1$ holds.
\end{proof}

Then we can obtain an $O(n^8)$-time and $O(n^6)$-space algorithm for solving {\CTLCSprob} based on Lemma~\ref{lem:main_idea}
(see Algorithm~\ref{alg:ctlcs}).
\begin{algorithm2e}[tbh]
  \caption{Algorithm for solving {\CTLCSprob}}
\label{alg:ctlcs}
\KwIn{Strings $S[1..n], T[1..n] \in \Sigma^*$}
    \KwOut{$\ctlcs(S,T)$}
    \SetKw{To}{to}
    \SetKw{And}{and}
    Find $i_1, i_2, \ldots, i_n$ that satisfy $S[i_1]>S[i_2]> \cdots >S[i_n]$\; \label{al1:i-order}
    Find $j_1, j_2, \ldots, j_n$ that satisfy $T[j_1]>T[j_2]> \cdots >T[j_n]$\; \label{al1:j-order}
    $\mathit{ctlcs} \leftarrow 0$\;
    \For{$i \leftarrow i_1$ \To $i_n$}{ \label{al1:for-i-order}
        \For{$j \leftarrow j_1$ \To $j_n$}{ \label{al1:for-j-order}
            \For{$\ell_1 \leftarrow 1$ \To $i$}{
                \For{$r_1 \leftarrow i$ \To $n$}{
                    \For{$\ell_2 \leftarrow 1$ \To $j$}{
                        \For{$r_2 \leftarrow j$ \To $n$}{
                            $\mathcal{C}_L \leftarrow 0$\;
                            \If{$\ell_1 \neq i$ \And $\ell_2 \neq j$}{
                                \For{$i' \leftarrow \ell_1$ \To $i-1$}{
                                    \For{$j' \leftarrow \ell_2$ \To $j-1$}{
                                        \If{$S[i']>S[i]$ \And $T[j']>T[j]$}{
                                            $\mathcal{C}_L \leftarrow \max(\mathcal{C}_L, C[i'][j'][\ell_1][i-1][\ell_2][j-1])$\;
                                        }
                                    }
                                }
                            }
                            $\mathcal{C}_R \leftarrow 0$\;
                            \If{$r_1 \neq i$ \And $r_2 \neq j$}{
                                \For{$i' \leftarrow i+1$ \To $r_1$}{
                                    \For{$j' \leftarrow j+1$ \To $r_2$}{
                                        \If{$S[i']>S[i]$ \And $T[j']>T[j]$}{
                                            $\mathcal{C}_R \leftarrow \max(\mathcal{C}_R, C[i'][j'][i+1][r_1][j+1][r_2])$\;
                                        }
                                    }
                                }
                            }
                            $C[i][j][\ell_1][r_1][\ell_2][r_2] \leftarrow \mathcal{C}_L+\mathcal{C}_R+1$\;
                        }
                    }
                }
            }
            $\mathit{ctlcs} \leftarrow \max(\mathit{ctlcs}, C[i][j][1][n][1][n])$\;
        }
    }
    \Return{$\mathit{ctlcs}$}
\end{algorithm2e}
Our algorithm computes a six-dimensional table $C$ 
for any $(i, j, \ell_1, r_1, \ell_2, r_2) \in [n]^6$ that satisfies $\ell_1 \leq i \leq r_1$, $\ell_2 \leq j \leq r_2$.
Notice that the processing order $i_1, i_2, \ldots, i_n$ (resp., $j_1, j_2, \ldots, j_n$) w.r.t. index $i$ (resp., $j$) 
has to satisfy $S[i_1] > S[i_2] > \cdots > S[i_n]$ (resp., $T[j_1] > T[j_2] > \cdots > T[j_n]$). 
The algorithm finally returns $\ctlcs(S,T) = \max\{C(i, j, 1, n, 1,n) \mid 1 \leq i \leq n, 1 \leq j \leq n\}$.
For each fixed $(i, j, \ell_1, r_1, \ell_2, r_2) \in [n]^6$ (i.e., $O(n^6)$ iterations), we can compute $\mathcal{C}_L$ and $\mathcal{C}_R$ in $O(n^2)$ time.
Therefore, we can compute table $C$ in $O(n^8)$ time and $O(n^6)$ space.
\begin{theorem}
\label{thm:ct-lcs}
  {\CTLCSprob} can be solved in $O(n^8)$ time and $O(n^6)$ space.
\end{theorem}

\subsubsection{$O(n^6)$-time $O(n^4)$-space algorithm}
\label{subsec:fast_alg}

In the sequel, we propose an improved algorithm that is based on the previous algorithm and runs in $O(n^6)$ time and $O(n^4)$ space.
The key observation is that $\mathcal{C}_L$ and $\mathcal{C}_R$ actually depend on only four variables.
Namely,
$\mathcal{C}_L$ depends on $(i, j, \ell_1, \ell_2)$, and $\mathcal{C}_R$ depends on $(i, j, r_1, r_2)$.
Let $L(i, j, \ell_1, \ell_2) = \max \mathcal{C}_L$ and $R(i, j, r_1, r_2) = \max \mathcal{C}_R$.
Then we can represent $C(i, j, \ell_1, r_1, \ell_2, r_2)$ as 
\[
    C(i, j, \ell_1, r_1, \ell_2, r_2) = L(i, j, \ell_1, \ell_2)+R(i, j, r_1, r_2)+1.
\]
Based on this recurrence, we can obtain the following alternative lemma.

\begin{lemma} \label{lem:main_idea_fast}
    For any $(i, j, \ell_1, r_1, \ell_2, r_2) \in [n]^6$ that satisfies $\ell_1 \leq i \leq r_1$, $\ell_2 \leq j \leq r_2$,
    the following recurrences hold:    
    \begin{linenomath}
    \begin{align*}
        L(i, j, \ell_1, \ell_2) = \max \{&L(i', j', \ell_1, \ell_2)+R(i', j', i-1, j-1)+1 \mid S[i']>S[i],\\
        &T[j']>T[j], \ell_1 \leq i' \leq i-1, \ell_2 \leq j' \leq j-1\} \cup \{0\},\\
R(i, j, r_1, r_2)=\max\{&L(i', j', i+1, j+1)+R(i', j', r_1, r_2)+1 \mid S[i']>S[i],\\
        &T[j']>T[j],\;i+1 \leq i' \leq r_1,\;j+1 \leq j' \leq r_2\} \cup \{0\}.
    \end{align*}
  \end{linenomath}
\end{lemma}

It follows from the definitions that 
\[
    \ctlcs(S, T)=\max\{L(i, j, 1, 1)+R(i, j, n, n)+1\ \mid 1 \leq i \leq n, 1 \leq j \leq n\}.
\]

Then we can obtain an $O(n^6)$-time and $O(n^4)$-space algorithm for solving {\CTLCSprob} based on Lemma~\ref{lem:main_idea_fast}.
Our algorithm computes two dimensional tables $L$ and $R$.
For each fixed $(i, j) \in [n]^2$ (i.e., $O(n^2)$ iterations), 
we can compute $L(i, j, \cdot, \cdot)$ and $R(i, j, \cdot, \cdot)$ in $O(n^4)$ time.
The algorithm finally returns $\ctlcs(S,T) = \max\{L[i][j][1][1]+R[i][j][n][n]+1 \mid 1 \leq i \leq n, 1 \leq j \leq n\}$.
Therefore, we can compute table $C$ in $O(n^6)$ time and $O(n^4)$ space.

\begin{theorem} \label{thm:ct-lcs_fast}
  {\CTLCSprob} can be solved in $O(n^6)$ time and $O(n^4)$ space.
\end{theorem}

We can compute a CT-LCS by storing the following additional information:
\ynnote*{modified}{We store pivot $(i',j')$ with $L(i, j, \ell_1, \ell_2)$ 
that satisfies $L[i][j][\ell_1][\ell_2] = L[i'][j'][\ell_1][\ell_2]+R[i'][j'][i-1][j-1]+1$ (also for $R$).
}If we do so, we can compute a CT-LCS by tracking back the tables from the pivot $(i,j)$ that gives $\ctlcs(S,T)$ in $O(|\ctlcs(S,T)|) = O(n)$ time.

\begin{corollary} \label{cor:trace_back}
    A CT-LCS can be computed in $O(n^6)$ time and $O(n^4)$ space.
\end{corollary}
 \subsection{{\CTLCSprob} algorithm for binary}
\label{sec:binary_LCS}

In this section, we propose an algorithm for solving {\CTLCSprob} for the binary alphabet $\{0, 1\}$.
Throughout this section, 
we assume that the strings $S$ and $T$ are binary strings
and discard the assumption that
all characters are distinct in $S$ and in $T$.

Similarly to our algorithm for {\CTMSeqprob} for binary alphabets, which was proposed in Section~\ref{sec:binary_matching},
our algorithm for {\CTMSeqprob} on binary strings is based on Lemma~\ref{lem:only1_ct_match} and Lemma~\ref{lem:bi_ct_match}.

Let $N_1(S)$ be the number of occurrences of $1$ in string $S$, 
and $L_{01}(S)$ the length of the longest non-decreasing subsequence
of $S$ that contains $0$.
If there is no such subsequence, let $L_{01}(S) = 0$.
We also define $\cand(S,T)$ as the maximum integer $k = |w|+i+j = |w|+i'+j'$ 
such that $w0^i1^j$ is a subsequence of $S$ and $w0^{i'}1^{j'}$ is a subsequence of $T$ 
for some string $w$, and integers $i, i' \geq 1$, $j, j' \geq 0$.
Then the following properties hold for subsequences $S'$ and $T'$ of $S$ and $T$ that give $\ctlcs(S,T)$.
\begin{itemize}
  \item If $0$ appears in both $S'$ and $T'$,
  $\ctlcs(S,T) = \cand(S,T)$ (by Lemma~\ref{lem:bi_ct_match}).
  \item If either $S'$ or $T'$ does not contain $0$,
  $\ctlcs(S,T)$ equals $\min(N_1(S), L_{01}(T))$ or $\min(N_1(T), L_{01}(S))$, respectively (from Lemma~\ref{lem:only1_ct_match}).
  \item If $0$ does not appears in both $S'$ and $T'$, $\ctlcs(S,T) = \min(N_1(S), N_1(T))$.
\end{itemize}

Due to the above properties, $\ctlcs(S,T) = \max(\cand(S,T),m_1, m_2, m_3)$ holds
for $m_1 = \min(N_{1}(S), L_{01}(T))$, $m_2 = \min(N_{1}(T), L_{01}(S))$, and $m_3 = \min(N_1(S), N_1(T))$.

Now we are ready to describe our algorithm.
Let $\lnd_S(i) = L_{01}(S[i..n])$ for any integer $i$ with $1 \leq i \leq n$.
For convenience, let $\lnd_S(n+1) = 0$.
Firstly, we compute $N_1(S[i..n])$, $N_1(T[i..n])$, $\lnd_S(i)$, and $\lnd_T(i)$ for all $i$ that satisfies $1 \leq i \leq n$.
We can easily compute $N_1(S[i..n])$ and $N_1(T[i..n])$ for all $i$ in $O(n)$ time and space.
We can also compute $\lnd_S(i)$ (and $\lnd_T(i)$ in a similar way) by using the following recurrence:
\begin{linenomath}
\begin{equation*}
  \lnd_S(i)= \left\{ 
    \begin{alignedat}{2}   
      &\max(\lnd_S(i+1)+1, N_1(S[i+1..n])+1) \quad & \text{if } S[i] = 0, \\   
      &\lnd_S(i+1)                     \quad & \text{if } S[i] = 1.
    \end{alignedat} 
  \right.
\end{equation*}
\end{linenomath}
These values can also be computed in $O(n)$ time and space.

Since 
$\cand(S,T)$ requires the length of $w$ which is described in the above discussion,
we use a data structure for computing the longest common subsequence $\lcs(i,j)$ of $S[1..i]$ and $T[1..j]$.
For convenience, we set $\lcs(i,0) = \lcs(0,i) = 0$ for all $i \in [n] \cup \{0\}$.
By using the Four-Russians method~\cite{MASEK198018},
we can compute an $O(n^2/\log n)$-space data structure in $O(n^2/\log n)$ time
that can answer $\lcs(i,j)$ in $O(\log^2 n)$ time for any $i \in [n]\cup \{0\}$ and $j \in [n]\cup \{0\}$.
By the definition of $\cand$, the following equation can be obtained:
\[
  \cand(S,T) = \max_{1 \leq \ell \leq n} \{\lcs(p_{\ell}-1, q_{\ell}-1) + \ell\}
\]
where $p_{\ell} = \max\{p \mid \lnd_S(p) = \ell \}$ and $q_{\ell} = \max\{q \mid \lnd_T(q) = \ell \}$
(see also Fig.~\ref{fig:ctlcs_prime} for an illustration).
\begin{figure}[tb]
  \centering
  \includegraphics[keepaspectratio, scale=0.5]{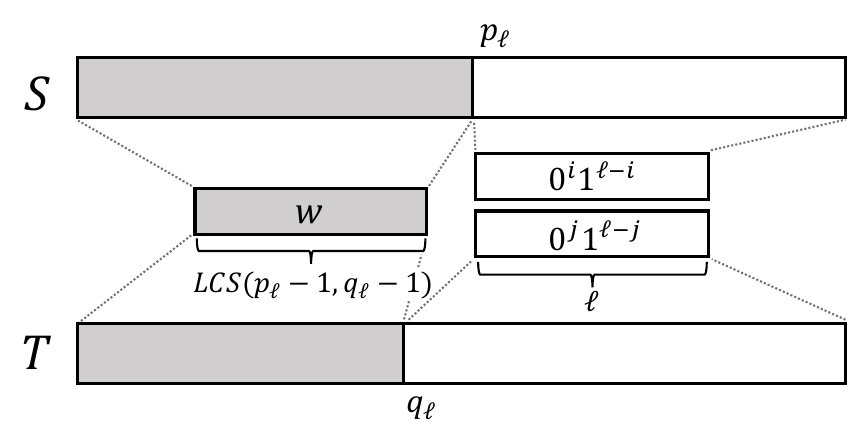}
  \caption{Illustration for our idea for computing $\cand(S,T)$.}
  \label{fig:ctlcs_prime}
\end{figure}Since we have already computed $\lnd_S(i)$, $\lnd_S(i)$, and the data structure for $\lcs$,
we can compute $\cand(S,T)$ in $O(n\log^2 n)$ time based on the above equation.
Finally, we can obtain $\ctlcs(S,T)$ by computing $\max(\cand(S,T),m_1, m_2, m_3)$ in constant time
(see also Algorithm~\ref{alg:ctlcs_binary}).\begin{algorithm2e}[tbh]
  \caption{Algorithm for solving {\CTLCSprob} on binary strings}
  \label{alg:ctlcs_binary}
\KwIn{Binary strings $S[1..n], T[1..n] \in \{0, 1\}^*$}
      \KwOut{$\ctlcs(S,T)$}
      \SetKw{To}{to}
      \SetKw{And}{and}
      Precompute data structure $\mathit{LCS}[i][j]$ that can answer $\lcs(i,j)$ in $O(\log^2 n)$ time
      for any $i \in [n]\cup \{0\}$ and $j \in [n]\cup \{0\}$\;
      $N_S[n+1] \leftarrow 0$\;
      \For{$i \leftarrow n$ \To $1$}{
          \If{$S[i] = 0$}{
              $N_S[i] \leftarrow N_S[i+1]$\;
          }
          \Else{
              $N_S[i] \leftarrow N_S[i+1]+1$\;
          }
      }
      $\lnd_S[n+1] \leftarrow 0$\;
      \For{$i \leftarrow n$ \To $1$}{
          \If{$S[i] = 0$}{
              $\lnd_S[i] \leftarrow \max(\lnd_S[i+1]+1, N_S[i+1]+1)$\;
          }
          \Else{
              $\lnd_S[i] \leftarrow N_S[i+1]$\;
          }
      }
      $p[\ell] \leftarrow 0$ for all $\ell \in [n]$\;
      \For{$i \leftarrow 1$ \To $n$}{
          $p[\lnd_S[i]] \leftarrow i$\;
      }
      Compute $N_T[i] = N_1(T[i..n])$, $\lnd_T[i] = \lnd_T(i)$, $q[\ell] = \max\{q \;|\; \lnd_T[q]=\ell\}$ 
      for all $i \in [n+1]$ and $\ell \in [n]$ in the same way\;
      $\mathit{cand} \leftarrow 0$\;
      \For{$\ell \leftarrow 1$ \To $n$}{
          \If{$p[\ell] \neq 0$ \And $q[\ell] \neq 0$}{
              $\mathit{cand} \leftarrow \max(\mathit{cand}, \mathit{LCS}[p[\ell]-1][q[\ell]-1]+\ell)$\;
          }
      }
      $m_1 \leftarrow \min(N_S[1], \lnd_T[1])$\;
      $m_2 \leftarrow \min(\lnd_S[1], N_T[1])$\;
      $m_3 \leftarrow \min(N_S[1], N_T[1])$\;
      \Return{$\max(\mathit{cand}, m_1, m_2, m_3)$}
\end{algorithm2e}

\begin{theorem}
\label{thm:ct-lcs_bi}
  {\CTLCSprob} on binary strings can be solved in $O(n^2/\log n)$ time and $O(n^2/\log n)$ space.
\end{theorem}

We can reconstruct a CT-LCS of $S$ and $T$ in $O(n\log n)$ time as follows:
If one of $m_1$, $m_2$, and $m_3$ gives $\ctlcs(S,T)$,
we can easily obtain a CT-LCS in $O(n)$ time by using $\lnd_S(i)$ and $\lnd_T(i)$.
Otherwise, two subsequences $S'$ and $T'$ which give $\ctlcs(S,T)$ can be represented as
$S' = w0^i1^{\ell - i}$ and $T' = w0^j1^{\ell - j}$ for some $i$ and $j$.
Integers $i$ and $j$ can be obtained by $\lnd_S(i)$ and $\lnd_T(i)$.
In the Four-Russians method, $(n \times n)$-table $\lcs$ is factorized into $(n/\log n \times n/\log n)$-blocks.
The data structure actually stores LCS values on boundaries of blocks.
Thus we can obtain string $w$ by tracing back in $O((n/\log n) \cdot \log^2n) = O(n\log n)$ time (see also Fig.~\ref{fig:lcs_trace}).
\begin{figure}[tb]
  \centering
  \includegraphics[keepaspectratio, scale=0.4]{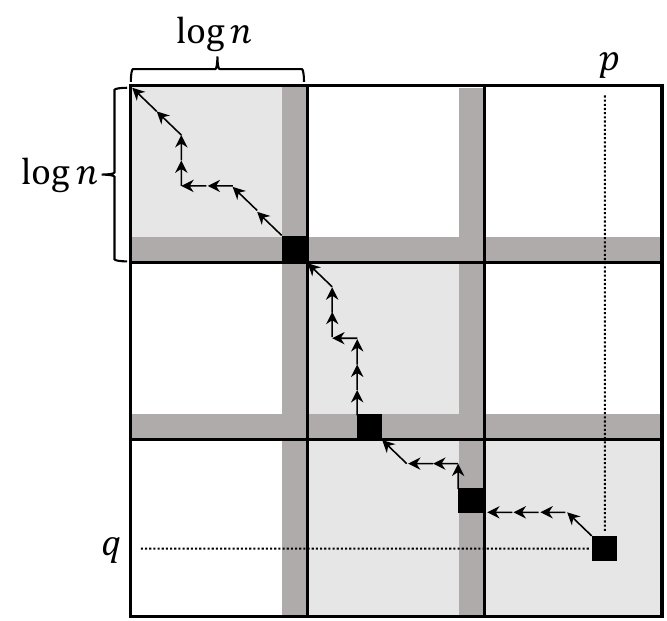}
  \caption{Illustration for an idea for reconstructing the LCS part.}
  \label{fig:lcs_trace}
\end{figure}

\begin{corollary}
  A CT-LCS of two binary strings can be computed in $O(n^2/\log n)$ time and $O(n^2/\log n)$ space.
\end{corollary}
 \clearpage
\subsection{Conditional lower bound for {\CTLCSprob}} \label{sec:lb_ctlcs}
In this section, we prove the following theorem. 
\begin{theorem}\label{th:ov2lcs}
  If there is an $\epsilon > 0$ such that {\CTLCSprob} for an alphabet of size $5$ 
  can be solved in $O((|S||T|)^{1-\epsilon})$ time, then OVH and SETH are false.
\end{theorem}

We extend the LCS reduction~\cite{BringmannK15} to our CT-subsequence version using a similar idea as described in Section~\ref{sec:lb_ctmseq}.
Let $A = \{\alpha_1, \ldots, \alpha_{n_A}\}$ and $B = \{\beta_1, \ldots, \beta_{n_B}\}$
be the input of the OV problem where $n_B \le n_A$ and $\alpha_i, \beta_j \in \{0, 1\}^d$ for every $1\le i\le n_A, 1\le j\le n_B$.
From $A$ and $B$,
we construct an input of {\CTLCSprob}:
string $S_A \in \Sigma^+$ of length $240n_Ad^2 + 56n_Ad + 100d^2 + 16n_A \in O(n_A\mathsf{poly}(d))$ and
string $T_B \in \Sigma^+$ of length $400n_Ad^2 + 25n_Bd^2 + 100d^3 + 14n_Bd + 100d^2 + 4n_B \in O((n_A + n_B)\mathsf{poly}(d))$,
where $\Sigma = \{\mt{0}, \mt{1}, \mt{2}, \mt{3}, \mt{4}\}$.

First, we define the \emph{coordinate gadgets} $C_1$ and $C_2$:
\begin{eqnarray*}
  C_1(a) &=& 
  \begin{cases}
    \mt{3434} & \mathrm{if} \; a = 0, \\
    \mt{4433} & \mathrm{if} \; a = 1.
  \end{cases} \\
  C_2(b) &=& 
  \begin{cases}
    \mt{4343} & \mathrm{if} \; b = 0, \\
    \mt{3333} & \mathrm{if} \; b = 1.
  \end{cases} 
\end{eqnarray*}
We can see that 
$\ctlcs(C_1(0), C_2(0)) = \ctlcs(C_1(0), C_2(1)) = \ctlcs(C_1(1), C_2(0)) = 3$
and $\ctlcs(C_1(1), C_2(1)) = 2$ hold.
Next, we construct \emph{vector gadgets} $V_1(\alpha)$ and $V_2(\beta)$ for all $\alpha\in A$ and $\beta\in B$:
\begin{alignat*}{12}
  V_1(\alpha) &={}& &\mt{2}^s \;C_1(\alpha[1])\;&\mt{2}^s&\;C_1(\alpha[2])\;&\mt{2}^s&\;\cdots\;&\mt{2}^s&\; C_1(\alpha[d])\; &\mt{2}^s&\;\mt{3434}\;&\mt{2}^s&, \\
  V_2(\beta)  &={}& &\mt{2}^s \;C_2(\beta[1])\;&\mt{2}^s&\;C_2(\beta[2])\;&\mt{2}^s&\;\cdots\;&\mt{2}^s&\; C_2(\beta[d])\; &\mt{2}^s&\;\mt{3333}\;&\mt{2}^s&,
\end{alignat*}
where $s = 5d$. 
Note that $|V_1(\alpha)| = |V_2(\beta)| = 5d^2 + 14d + 4$.
The following Lemma~\ref{lem:vg_ctlcs} holds. \begin{lemma} \label{lem:vg_ctlcs}
  $\ctlcs(V_1(\alpha), V_2(\beta)) = D + 1$ holds if $\alpha$ and $\beta$ are orthogonal, 
  and otherwise $\ctlcs(V_1(\alpha), V_2(\beta)) \le D$ holds, where $D = s(d+2) + 3(d+1) - 1$.
\end{lemma}
\begin{proof}
  For each $1 \le i \le d$, Cartesian trees of $C_1(\alpha[i])$ and $C_2(\beta[i])$ belong to the left subtree of the $(i+1)$th substring $\mt{2}^s$ 
  in $CT(V_1(\alpha))$ and $CT(V_2(\beta))$, respectively,
  since $C_1(\alpha[i])$ and $C_2(\beta[i])$ consist only of $\mt{3}$ and $\mt{4}$.
Now we consider the following two cases:
  \begin{enumerate}      
    \item If $\alpha$ and $\beta$ are orthogonal, 
      $\ctlcs(C_1(\alpha[i]), C_2(\beta[i])) =  3$ holds for all $1 \leq i \leq d$.
      Therefore, by assigning every $\mt{2}^s$ in $V_2(\beta)$ to every $\mt{2}^s$ in $V_1(\alpha)$, and
      $C_2(\beta[i])$ to $C_1(\alpha[i])$ for each $1 \le i \le d$,
      we obtain a common CT-subsequence of length $D+1$.
      \tmnote*{changed}{Assume on the contrary that there exists a common CT-subsequence $S$ of $V_1(\alpha)$ and $V_2(\beta)$ with $|S| \ge D+2$.
Since the total number of $\mt{3}$'s and $\mt{4}$'s in $V_1(\alpha)$ (or $V_2(\beta)$) is $4(d+1)$,
        any subsequence of $V_1(\alpha)$ (or $V_2(\beta)$) that CT-matches $S$ contains at least $D+2 - 4(d+1)= 5d(d+1) + 4d$ $\mt{2}$'s.
        This implies that such subsequences include at least one $\mt{2}$ from each run $\mt{2}^s = \mt{2}^{5d}$ in the vector gadgets.
        Thus, their Cartesian-trees form comb-like structures with the frontier consisting only of $\mt{2}$'s. 
        Therefore, we have to assign $C_1(\alpha[i])$ to $C_2(\beta[i])$ for each $1 \leq i \leq d$.
However, such an assignment results in a common CT-subsequence of length at most $D+1$ from the above discussion, a contradiction.
      }\item Otherwise, there exists some $k$ with $\alpha[k] = \beta[k] = 1$,
      and $\ctlcs(C_1(\alpha[k]), C_2(\beta[k])) = 2$ holds.
      \tmnote*{changed}{Assume on the contrary that there exists a common CT-subsequence $S$ of $V_1(\alpha)$ and $V_2(\beta)$ with $|S| \ge D+1$.
A subsequence of $V_1(\alpha)$ (or $V_2(\beta)$) that CT-matches $S$ 
        contains at least $D+1 - 4(d+1)= 5d(d+1) + 4d - 1$ $\mt{2}$'s. 
        From a similar discussion to the case $\alpha$ and $\beta$ are orthogonal,
        we have to assign $C_1(\alpha[i])$ to $C_2(\beta[i])$ for each $1 \leq i \leq d$.
However, such an assignment results in a common CT-subsequence of length at most $D$ due to $\ctlcs(C_1(\alpha[k]), C_2(\beta[k])) = 2$,
        a contradiction.
      }\end{enumerate}
\end{proof}

Also, we construct \emph{normalized vector gadgets} 
$N_1(\alpha)$ and $N_2(\beta)$ for all $\alpha\in A$ and $\beta\in B$,
joining dummy vector gadget $V_{\zero} ={} \mt{2}^s C_1(\zero[1]) \mt{2}^s \cdots \mt{2}^s C_1(\zero[d]) \mt{2}^s \mt{4433} \mt{2}^s$:
\begin{alignat*}{12}
  N_1(\alpha) &={} V_1(\alpha) \hspace{3mm} \mt{1}^t &V_{\zero}, \\
  N_2(\beta)  &={} \hspace{2.7mm} \mt{1}^t \hspace{3.4mm} V_2(\beta) \hspace{1.5mm} &\mt{1}^t,
\end{alignat*}
where $t = 10d^2$.
Note that $|N_1(\alpha)| = 20d^2 + 28d + 8$, $|N_2(\beta)| = 25d^2 + 14d + 4$ 
and $\ctlcs(V_{\zero}, V_2(\beta)) = D$.
For these gadgets, the following lemma also holds.
\begin{lemma} \label{lem:nvg_ctlcs}
  $\ctlcs(N_1(\alpha), N_2(\beta)) = E + 1$ holds if $\alpha$ and $\beta$ are orthogonal, 
  and otherwise $\ctlcs(N_1(\alpha), N_2(\beta)) = E$ holds, where $E = t + D$.
\end{lemma}
\begin{proof}
  We consider an alignment of $N_1(\alpha)$ in $N_2(\beta)$.
  \begin{enumerate}
    \item If $\alpha$ and $\beta$ are orthogonal, 
      \tmnote*{changed}{we align $V_1(\alpha)$ to $V_2(\beta)$ and $\mt{1}^t$ in $N_1(\alpha)$ to the right $\mt{1}^t$ in $N_2(\beta)$.
        From such an alignment, we can construct a common CT-subsequence of $N_1(\alpha)$ and $N_2(\beta)$ of length $E+1$
by Lemma~\ref{lem:vg_ctlcs}.
        Assume on the contrary that there exists a common CT-subsequence $S$ of $N_1(\alpha)$ and $N_2(\beta)$ with $|S| > E+1$.
        From a similar discussion to Lemma~\ref{lem:vg_ctlcs},
        any subsequence of $N_2(\beta)$ that CT-matches $S$ includes at least $10d^2-d$ $\mt{1}$'s (from at least one block of $\mt{1}^t$).
Thus, the Cartesian-tree of such a subsequence of $V_2(\beta)$ satisfies any of the following:
      }(1) it forms a comb-like structure with the frontier consisting only of $\mt{1}$'s \ynnote*{added}{(i.e., a subsequence of form $\mt{1}^i x \mt{1}^j$ for some integers $i$, $j$ and some string $x$)}, 
      (2) the characters in $V_2(\beta)$ belongs to some $\mt{1}$ in the left  $\mt{1}^t$ in $N_2(\beta)$ (i.e., a subsequence of form $\mt{1}^i x$), or
      (3) the root of the Cartesian-tree is         some $\mt{1}$ in the right $\mt{1}^t$ in $N_2(\beta)$ and
          the characters in $V_2(\beta)$ belongs to some $\mt{1}$ in the right $\mt{1}^t$ in $N_2(\beta)$ (i.e., a subsequence of form $x \mt{1}^i$).
In the case (1), if we align some $\mt{1}$'s in the left/right $\mt{1}^t$ to the $\mt{1}^t$ in $N_1(\alpha)$,
      we must align the all $\mt{1}$'s in the right/left $\mt{1}$ in $N_2(\beta)$ to the characters $c$ in $V_{\zero}$ / $V_1(\alpha)$, respectively.
      In this condition, the subsequence can not include other $c$ in $V_{\zero}$ / $V_1(\alpha)$.
      For any $c$ in $\{\mt{2}, \mt{3}, \mt{4}\}$, the length of the possible subsequence is clearly less than $E+1$ 
      even if all $\mt{1}$'s in $N_2(\beta)$ can be aligned to $N_1(\alpha)$.
In the case (2), we must align the all $\mt{1}$'s in the left $\mt{1}^t$ in $N_2(\beta)$ to the characters in $\mt{1}^t \; V_{\zero}$, 
      or to the characters only in $V_1(\alpha)$ and $V_{\zero}$.
      Then, from similar discussion to the case (1), the length of the possible subsequence is less than $E$.
In the case (3), we must align the all $\mt{1}$'s in the right $\mt{1}^t$ in $N_2(\beta)$ to the characters in $V_1(\alpha)$,
      or to the characters only in $V_1(\alpha)$ and $V_{\zero}$.
      Then, from similar discussion to the case (1), the length of the possible subsequence is less than $E+1$.
    \item Otherwise, we can align $V_{\zero}$ to $V_2(\beta)$ and $\mt{1}^t$ in $N_1(\alpha)$ to the left $\mt{1}^t$ in $N_2(\beta)$.
      From a similar discussion to the above case and Lemma~\ref{lem:vg_ctlcs}, $\ctlcs(N_1(\alpha), N_2(\beta)) = E$.
  \end{enumerate}
\end{proof}

Now we construct $S_A$ of length $240n_Ad^2 + 56n_Ad + 100d^2 + 16n_A \in O(n_A\mathsf{poly}(d))$ and $T_B$ of length $400n_Ad^2 + 25n_Bd^2 + 100d^3 + 14n_Bd + 100d^2 + 4n_B \in O((n_A + n_B)\mathsf{poly}(d))$ as follows:
\begin{eqnarray*}
  S_A &={}& \mt{0}^u \; (N_1(\alpha_1) \; \mt{0}^u \; N_1(\alpha_2) \; \mt{0}^u \; \cdots \; \mt{0}^u \; N_1(\alpha_{n_A}) \; \mt{0}^u )^2 \\
  T_B &={}& \mt{0}^{u(2n_A + 1)} \; N_2(\beta_1) \; \mt{0}^u \; N_2(\beta_2) \; \mt{0}^u \; \cdots \; \mt{0}^u \; N_2(\beta_{n_B}) \; \mt{0}^{u(2n_A + 1)},
\end{eqnarray*}
where $u = 100d^2$.

We show the following lemmas:
\begin{lemma}\label{lem:seq_gadgets_1}
  If there exist $\alpha \in A$ and $\beta \in B$ that are orthogonal,
  then $\ctlcs(S_A, T_B) \ge F + 1$ holds, where $F = En_B + u(2n_A + 1)$.
\end{lemma}
\begin{proof}
  In this proof, we give an alignment of $T_B$ to $S_A$ 
  from which a common CT-subsequence of length $F + 1$ can be constructed.	
  Let $\alpha_i \in A$ and $\beta_j \in B$ be vectors which are orthogonal. 
  If $j \leq i$, we align $N_2(\beta_j)$ to the \emph{first} $N_1(\alpha_i)$ in $S_A$.
  By doing so, it is guaranteed that
  every normalized vector gadget $N_2(\beta_x)$ in $T_B$~($x \ne j$) can be aligned to
  some normalized vector gadget in $S_A$. 
  Symmetrically, if $j > i$, we align $N_2(\beta_j)$ to the \emph{second} $N_1(\alpha_i)$ in $T_A$.
  Then, we align all the normalized vector gadgets in $T_B$, except $N_2(\beta_j)$, 
  be aligned to normalized vector gadgets to the right and left of $N_1(\alpha_i)$ 
  in order from closest to $N_1(\alpha_i)$, and
$u(2n_A+1)$ $\mt{0}$'s in $T_B$ are aligned to $\mt{0}$'s which interleave normalized vector gadgets in $S_A$. The length of the common CT-subsequence constructed from such alignment 
  is at least $E + 1 + (n_B-1)E + u(2n_A+1) = F + 1$ by Lemma~\ref{lem:nvg_ctlcs}.
\end{proof}

\begin{lemma}\label{lem:seq_gadgets_2}
  If $\alpha$ and $\beta$ are not orthogonal for all pairs of $\alpha \in A$ and $\beta \in B$,
  then $\ctlcs(S_A, T_B) \le F$ holds.
\end{lemma}
\begin{proof}
\ynnote*{changed}{Since $\alpha$ and $\beta$ are not orthogonal, 
any alignment between two normalized vector gadgets has length at most $E$ by Lemma~\ref{lem:nvg_ctlcs}.
By a similar alignment of the proof of Lemma~\ref{lem:seq_gadgets_1}, 
  we can construct a common CT-subsequence of length $En_B + u(2n_A+1) = F$ from the condition,
  such that we align all the normalized vector gadgets in $T_B$ 
  be aligned to normalized vector gadgets in $S_A$
  in order, 
  and $u(2n_A+1)$ $\mt{0}$'s in $T_B$ are aligned to $\mt{0}$'s which interleave normalized vector gadgets in $S_A$.
}Here, we prove such alignment, which we denote by $M$, is optimal.
  \ynnote*{modified}{Assume on the contrary that there exists a subsequence $S'$ of $S_A$ that CT-matches a common CT-subsequence of $T_B$ with $|S'| \ge F+1$.
  }Since the total length of normalized vector gadgets in $S_A$ is $40n_Ad^2 + 56n_Ad + 16n_A$,
  $S'$ must include at least $F - (40n_Ad^2 + 56n_Ad + 16n_A) > 88n_Ad^2 + 100d^2 > u$ $\mt{0}$'s.
We consider the following two cases:
\begin{enumerate}
    \item \label{case:1}
      If the common CT-subsequence includes all $\mt{0}$'s in $S_A$, 
      since the Cartesian-tree of such subsequences forms a comb-like structure 
      (the frontier consists of $\mt{0}$'s
      and each normalized vector gadget belongs to the left child of $\mt{0}$),
      we can not align a normalized vector gadget in $T_B$ to two or more normalized vector gadgets
      in $S_A$ keeping the isomorphism.
      \ynnote*{modified}{There can be a possible alignment of the reverse direction (i.e., we can align a single normalized vector gadget in $S_A$ to multiple normalized vector gadgets in $T_B$),
        but the alignment
is inefficient keeping the isomorphism since non-aligned normalized vector gadgets in such alignment increase more than that in $M$ (recall that $n_B \le n_A$).
      More formally, if we align a single $N_1(\alpha)$ to $k$ normalized vector gadgets in $T_B$, the possibility of the length of common subsequence is $|N_1(\alpha)|$ since 
      $|N_1(\alpha)| = 20d^2 + 28d + 8 < k|N_2(\beta)| = k(25d^2 + 14d + 4)$ for any $k \ge 2$.
      In the above (optimal) alignment, 
      we have a length-$kE$ common subsequence by Lemma~\ref{lem:nvg_ctlcs}, 
      but $kE = k(15d^2+13d+2) > |N_1(\alpha)|$ holds for any $k$ and $d$.
      }Thus, $M$ is optimal.
\item \label{case:2}
      Otherwise, we drop some $\mt{0}$'s in $S_A$ (but such subsequences contains one or more $\mt{0}$'s).
      Here, we consider the following two cases w.r.t. the structure of Cartesian-tree of the subsequences:
      \begin{enumerate}
        \item If the subsequences include one or more $\mt{0}$'s in each $\mt{0}^u$ in $S_A$, 
          the Cartesian-tree of the subsequences forms a comb-like structure.
          Thus, from a similar discussion to the case~\ref{case:1}, 
          we have to align the normalized vector gadgets in $S_A$ and $T_B$ one to one 
          for the longest common CT-subsequence. 
          Then, the length of the subsequences constructed such alignment is clearly less than $F+1$.
        \item Otherwise, the Cartesian-tree of the subsequences forms the "mixed comb-like" structures 
          which include some mixed normalized vector gadgets belonging to the left child of $\mt{0}$, 
          and include less than a mixed normalized vector gadget belonging to the right child of the rightmost $\mt{0}$. 
          In such subsequences, in the case we mix two normalized vector gadgets belonging to the left child of $\mt{0}$, we drop a $\mt{0}^u$.
          However, even if we can align all characters in two normalized vector gadgets to $T_B$, 
          the length of such subsequence is less than the length of
          subsequences which include $\mt{2}^n$ between two normalized vector gadgets
          since $2|N_1(\alpha)| < u$ holds. 
          The similar discussion can be applied to the case we mix three or more normalized vector gadgets
          since $(x+1)|N_1(\alpha)| < xu$ holds for integer $x > 2$.
          Also, the similar discussion can be applied to the case we mix two or more normalized vector gadgets belonging to the right child of $\mt{0}$,
          since $x|N_1(\alpha)| < xu$ holds for integer $x > 2$.
          Therefore, the length of the subsequences constructed mixed comb-like structures is less than the length of the subsequences
          constructed comb-like structures (that is, less than $F+1$).
      \end{enumerate}
  \end{enumerate}
\end{proof}

From Lemmas~\ref{lem:seq_gadgets_1} and~\ref{lem:seq_gadgets_2}, 
we have a reduction from the OV problem to {\CTLCSprob} such that
$S_A \in O(n_A\mathsf{poly}(d))$ and $T_B \in O((n_A + n_B)\mathsf{poly}(d))$.
If {\CTLCSprob} can be solved in $O((|S_A||T_B|)^{1-\epsilon})$ time, 
then the OV problem can be solved in $O((n_A^2\mathsf{poly}(d))^{1-\epsilon})$ time, 
which contradicts OVH (and SETH).
Therefore, Theorem~\ref{th:ov2lcs} holds.

 \section{Conclusions and future work}
\label{sec:concl}

This paper studied the two problems, {\CTMSeqprob} and {\CTLCSprob},
that relate to
subsequences matching under the Cartesian-tree equivalence.
\tmnote*{changed}{We showed that quadratic and weakly subquadratic solutions
exist for {\CTMSeqprob} and {\CTLCSprob}, respectively, in the case of binary alphabets.
}We also presented strongly subquadratic solutions are unlikely to exist
for both {\CTMSeqprob} and {\CTLCSprob} for alphabets of sizes 4 and 5, respectively.
In addition, we gave a polynomial time solution
to {\CTLCSprob} in the case of general ordered alphabets.

A large gap remains between
the $O(n^6)$-time upper bound and the $O(n^{2-\epsilon})$-time lower bound
for {\CTLCSprob} in the general case,
and how to close this gap is an intriguing open question.
Another open question is to close the gaps in
the alphabet sizes for the subquadratic time complexities of {\CTMSeqprob} and {\CTLCSprob}.

\section*{Acknowledgments}
The authors thank Shay Mozes, Oren Weimann, and Tsubasa Oizumi for discussions on subsequence CT-matching in the binary case.

This work was supported by JSPS KAKENHI Grant Numbers
JP23H04381, JP24K20734~(TM),
JP21K17705, JP23H04386~(YN),
JP20H05964, JP23K24808~(SI).

\bibliography{ref}
\end{document}